\documentclass[12pt]{article}
\usepackage[a4paper]{geometry}
\usepackage{babel}

\usepackage{fancyhdr}
\usepackage{amsmath}
\usepackage{amsfonts}
\usepackage{amstext}
\usepackage{amssymb}
\usepackage{amsthm}
\usepackage{amscd}

\usepackage{soul}
\usepackage{comment}

\usepackage[pagebackref,draft=false]{hyperref}
\hypersetup{colorlinks,
linkcolor=myrefcolor,
citecolor=mycitecolor,
urlcolor=myurlcolor}

\usepackage[capitalize]{cleveref}
\usepackage{caption}
\usepackage{etaremune}

\usepackage{xcolor}
\definecolor{myurlcolor}{rgb}{0,0,0.4}
\definecolor{mycitecolor}{rgb}{0,0.5,0}
\definecolor{myrefcolor}{rgb}{0.5,0,0}
\usepackage{graphicx}
\usepackage{tikz}
\usepackage{tikz-cd}
\usepackage{mathrsfs}

\usepackage{etoolbox}
\usepackage{makeidx}
\usepackage{sectsty}
\usepackage{dsfont}
\usepackage{enumitem} 
\usepackage[]{latexsym}
\usepackage{braket}
\usepackage{caption}
\usepackage[utf8]{inputenx}
\usepackage[T1]{fontenc}
\usepackage{lmodern}
\usepackage{textcomp}
\usepackage{microtype}
\usepackage{totcount}
\usepackage{blindtext}
\newtheorem{remark}{Remark}

\newtheorem{theorem}{Theorem}

\newtheorem{proposition}{Proposition}
\newtheorem{definition}{Definition}

\newtheorem*{proof*}{Proof}

\newcommand{\be}{\begin{equation}}
\newcommand{\ee}{\end{equation}}
\newcommand{\bea}{\begin{eqnarray}}
\newcommand{\eea}{\end{eqnarray}}
\newcommand{\vsp}{\vspace{0.4cm}}
\newcommand{\grit}[1]{{\bfseries {\itshape {#1}}}}

\newcommand{\ra}{\rightarrow}

\newcommand{\hh}{\mathcal{H}}
\newcommand{\bh}{\mathcal{B}(\mathcal{H})}
\newcommand{\Glh}{\mathcal{GL}(\mathcal{H})}

\newcommand{\Uh}{\mathcal{U}(\mathcal{H})}

\newcommand{\SUh}{\mathcal{SU}(\mathcal{H})}

\newcommand{\Tr}{\textit{Tr}}

\newcommand{\stsp}{\mathscr{S}}

\newcommand{\stav}{\mathscr{V}}

\newcommand{\pos}{\mathscr{P}}

\newcommand{\gr}{\mathrm{g}}

\newcommand{\GG}{\mathrm{G}}

\title{Monotone metric tensors in \\ Quantum Information Geometry}

\author{F. M. Ciaglia$^{1,6}$ \href{https://orcid.org/0000-0002-8987-1181}{\includegraphics[scale=0.7]{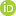}}, F. Di Cosmo$^{1,2,7}$  \href{https://orcid.org/0000-0002-8987-1181}{\includegraphics[scale=0.7]{ORCID.png}}, F. Di Nocera$^{3,8}$\href{https://orcid.org/0000-0002-1415-2422}{\includegraphics[scale=0.7]{ORCID.png}}, P. Vitale$^{4,5,9}$\href{https://orcid.org/0000-0002-5146-410X}{\includegraphics[scale=0.7]{ORCID.png}}\\
\footnotesize{$^{1}$\textit{ Depto. de Matem\'aticas, Univ. Carlos III de Madrid, Legan\'es, Madrid, Spain}} \\
\footnotesize{$^{2}$\textit{ ICMAT, Instituto de Ciencias Matem\'{a}ticas (CSIC-UAM-UC3M-UCM)}}  \\
\footnotesize{$^{3}$\textit{ Max Planck Institute for Mathematics in the Sciences, Leipzig, Germany}} \\
\footnotesize{$^{4}$\textit{ Dipartimento di Fisica ``E. Pancini'', Universit\`a di Napoli Federico II,  Naples, Italy}} \\
\footnotesize{$^{5}$\textit{ INFN-Sezione di Napoli, Naples, Italy}} \\
\footnotesize{$^{6}$\textit{ e-mail:   \texttt{ fciaglia[at]math.uc3m.es }}}\\
\footnotesize{$^{7}$\textit{ e-mail: \texttt{fcosmo[at]math.uc3m.es}}} \\
\footnotesize{$^{8}$\textit{ e-mail: \texttt{fabiodncr[at]gmail.com} and \texttt{dinocer[at]mis.mpg.de}}} \\
\footnotesize{$^{9}$\textit{ e-mail: \texttt{vitale[at]na.infn.it}}} 
}

\begin{document}

\maketitle

\begin{abstract}
We review some geometrical aspects pertaining to the world of monotone quantum metrics in finite dimensions. 
%%%%
Particular emphasis is given to an unfolded perspective for quantum states that is built out of the spectral theorem and is naturally suited to investigate the comparison with the classical case of probability distributions.
%%%%%
\end{abstract}

\tableofcontents

\thispagestyle{fancy}

\section{Introduction}

The investigation of the geometrical aspects of the space of quantum states is a well-established subject of theoretical and mathematical physics.
%%%%%
In particular, given the recent interest in quantum information theory stemming from the possible applications of quantum technologies, and given the success of the application of the methods of the so-called information geometry to diverse fields like estimation theory, hypothesis testing, and machine learning, the investigation of the differential geometric properties of the space of a finite-level quantum system, especially those related with aspects of Riemannian geometry, is an incredibly active field of research.
%%%%%

The aim of this work is precisely to review (some aspects of) what is known about the Riemannian geometry of the space of quantum states of a finite-level system as seen from a recently introduced point of view.  The latter  is based on the unfolding of a quantum state $\rho$ into a couple $(\vec{p},\mathbf{U})$, where $\vec{p}$ is the probability vector of eigenvalues of $\rho$, and $\mathbf{U}$ is a suitable unitary operator diagonalizing $\rho$.
%%%%%
The appearance of a probability vector immediately calls for a parallel with the classical case where probability vectors play the role of quantum states, and where the Fisher-Rao metric tensor determines the relevant Riemannian geometry. 
%%%%%
Indeed, we will  see that the unfolded perspective actually allows for a direct and clear understanding of the Riemannian metric tensors used in quantum information geometry in terms of a purely classical-like contribution and a quantum contribution.
%%%%
Quite interestingly, the classical-like contribution coincides with the Fisher-Rao metric tensor mentioned before.
%%%%
Regarding the purely quantum contribution, it can be essentially seen as a  “restriction” to the orbits of the special unitary group of a weighted version of the Cartan-Killing form, whose weights depend on the eigenvalues of the quantum states.
%%%%%

The work is structured as follows.
%%%%
In section \ref{sec: geometry of quantum states}, we will review the differential geometric properties of the space of quantum states that will be needed throughout the rest of the paper.
%%%%
In particular, we will briefly review the partition of the space of quantum states into the disjoint union of orbits of a nonlinear action of the general linear group of the Hilbert space of the system, and we will devote some time in reviewing three alternative identifications of tangent vectors to faithful quantum states, namely, the Jordan, the square-root, and the exponential identification.
%%%%
These identifications of tangent vectors seem to be seldom used in the current literature despite their usefulness in connection with three different monotone metric tensors as argued in subsections \ref
{subsec: B-H metric tensor}, \ref{subsec: W-Y metric tensor}, and \ref{subsec: B-K-M metric tensor}.
%%%%%

The purpose of section \ref{sec: Cencov and Petz} is to review Petz's classification of all the Riemannian geometries on the manifold of faithful quantum states  which are monotone with respect to the action of completely-positive, trace-preserving maps.
%%%%
This family of metric tensors determines all the Riemannian geometries on the manifold of faithful quantum states that are relevant from the point of view of quantum information geometry, very much like the Fisher-Rao metric tensor determines the Riemannian geometry of classical probability distributions.
%%%%

In section \ref{sec: from rel entropies to mono metrics},  we discuss the deep relation existing between quantum relative entropies and the monotone metric tensors reviewed in the previous section. 
%%%%%
The purely mathematical aspects of this relation are then thoroughly analysed section \ref{sec: from 2-functions to covariant tensors},  which contains  a coordinate-free approach to the algorithm needed to extract a Riemannian metric tensor from a relative entropy function.
%%%%%

Section \ref{sec: unfolding of quantum states}  represents the main part of the work. It is    devoted to the introduction of the unfolding procedure for quantum states alluded to before, and to  its application to the family of monotone metric tensors reviewed in section \ref{sec: Cencov and Petz}. Moreover,  the extraction algorithm presented in section \ref{sec: from 2-functions to covariant tensors} shall be   applied to a suitable unfolded version of the family of relative g-entropies introduced in section  \ref{sec: from rel entropies to mono metrics}.
%%%%
What emerges from the results presented in this section is a clear picture where the same classical-like contribution appears in the unfolding of every monotone metric tensor, and is found to coincide exactly  with the Fisher-Rao metric tensor.
%%%%%%
Moreover, this classical-like contribution is also responsible for the appearance of a family of universal geodesics which are common to all monotone metric tensors and, as argued in subsection \ref{subsec: universal geodesics}, it can be visualised as the geodesics of the Fisher-Rao metric tensor, properly immersed in the manifold of faithful quantum states.
%%%%%

Finally, in section \ref{sec: conclusions}, we share some concluding remarks on the work done and on some of its possible future applications.
%%%%%

\section{Geometrical aspects of the space of quantum states}\label{sec: geometry of quantum states}

In this section we  introduce the \emph{space of quantum states} and those of its geometrical aspects that  are more relevant to quantum information geometry.
%%%
In the standard picture of quantum mechanics, a quantum system is described by means of mathematical structures that are built out of a complex Hilbert space $\hh$ which identifies the system under investigation.
%%%%
For instance, (bounded) observables are associated with self-adjoint elements in $\bh$, and quantum states are associated with suitable normalized, positive linear functionals on observables.
%%%%
Roughly speaking, the relation between quantum states and observables in quantum mechanics is analogous to the relation between real-valued random variables and probability distributions in classical probability.
%%%%
Let us elaborate a little bit on that.
%%%%
Let us consider a probability space $(\mathcal{X},\mu)$ associated with a classical probabilistic system, and a quantum system with Hilbert space $\hh$.
%%%%
Real-valued random variables on $(\mathcal{X},\mu)$ form an algebra just like self-adjoint operators on $\hh$, however, the latter, when endowed with the anti-commutator product
\be\label{eqn: Jordan product}
\{\mathbf{a},\mathbf{b}\}\,:=\,\frac{1}{2}\left(\mathbf{ab} + \mathbf{ba}\right)
\ee
coming from the associative product in $\bh$,  form a non-associative Jordan algebra rather than an associative algebra like random variables on $(\mathcal{X},\mu)$.
%%%%%
This instance is at the hearth of the differences between the classical and quantum realms.
%%%%
Every probability distribution on $(\mathcal{X},\mu)$ which is absolutely continuous with respect to $\mu$ can be identified with a real-valued function $p$ in $\mathcal{L}^{1}(\mathcal{X},\mu)$ which is non-negative and is such that its integral on the whole $\mathcal{X}$ is 1.
%%%%
Every such probability distribution determines a linear functional on random variables by means of integration.
%%%%
Specifically, given   $p$ and the real-valued random variable $f$, we compute the mean value $\langle f\rangle_{p}$ according to
\be
\langle f\rangle_{p}\,:=\,\int_{\mathcal{X}}\,f\,p\,\mathrm{d}\mu,
\ee
and it is clear that the mean value is a linear functional on random variables.
%%%%
Since $p$ is a probability density function, the mean value functional will take non-negative values whenever the random variable $f$ is non-negative, and will give $1$ when applied to the identity function.
%%%%
If we focus only on random variables which are essentially bounded with respect to $\mu$, that is, elements of $\mathcal{L}^{\infty}(\mathcal{X},\mu)$, we see that the mean value functional takes only finite values.
%%%%

In the quantum context, random variables are replaced by self-adjoint operators in $\bh$, where $\hh$ is the Hilbert space  of the system, and probability distributions are replaced by quantum states, i.e.,  trace-class operators in $\bh$ which are positive semidefinite and have unit trace.
%%%%
Recall that an operator $\rho\in\bh$ is called positive semidefinite, written $\rho\geq 0$, if $\langle\psi|\rho(\psi)\rangle\geq 0$ for all $\psi\in\hh$.
%%%%
Then, the quantum counterpart of the mean value functional reads
\be\label{eqn: quantum mean value}
\langle\mathbf{a}\rangle_{\rho}\,:=\,\Tr_{\hh}(\rho\,\mathbf{a}),
\ee
where $\Tr_{\hh}$ is the standard Hilbert space trace, and we may look at $\Tr_{\hh}$ as the quantum analogue of the measure $\mu$ in the classical case.
%%%%
The fact that $\rho$ is positive semidefinite implies that the mean value functional in equation \eqref{eqn: quantum mean value} takes non-negative values on every quantum observable $\mathbf{a}$ which is itself positive semidefinite.
%%%%
Note that, however, a positive semidefinite observable $\mathbf{a}$ need not be trace-class just as, in the classical case, a positive random variable in $\mathcal{L}^{\infty}(\mathcal{X},\mu)$ need not be in $\mathcal{L}^{1}(\mathcal{X},\mu)$ (but, of course, this instance becomes relevant only in the infinite-dimensional case).
%%%%
 
Motivated by the previous discussion, we give the following definition

\begin{definition}
Given a  quantum system with Hilbert space $\hh$, let us denote with $\mathcal{B}_{tc}(\hh)$ the space of bounded, trace-class linear operators on $\hh$.
%%%
Then, the space of quantum states $\overline{\stsp}(\hh)$ is defined as
\be
\overline{\stsp}(\hh):=\left\{\rho\in\mathcal{B}_{tc}(\hh)\,|\;\;\rho^{\dagger}=\rho,\;\;\rho\geq 0,\;\Tr_{\hh}(\rho)=1\right\}.
\ee
\end{definition}
%%%%
The space of quantum states is a convex set and it turns out it is also compact.
%%%%
However, when $\hh$ is infinite-dimensional, the topology in which $\overline{\stsp}(\hh)$ is convex is not the norm topology on $\bh$ \cite[p. 53]{B-R-1987-1}. 
%%%%

Up to now, there is not a unique, satisfactory theory of quantum information geometry in the infinite-dimensional case, and some attempts may be found, for instance, in \cite{A-V-2005,C-I-J-M-2019,C-L-M-1983,C-M-P-1990,G-K-M-S-2018,G-S-2000,Jencova-2006,Streater-2004}.
%%%%%
Therefore, throughout the paper, we will work only with finite-level quantum systems for which the associated Hilbert space is  finite-dimensional, and for which a satisfactory formulation of quantum information geometry is possible.
%%%%
In this case, it is not hard to see that, when we look at  $\overline{\stsp}(\hh)$  as a convex subset of the hyperplane $\mathcal{B}_{sa}^{1}(\hh)$ of self-adjoint elements in $\bh$ with unit trace, it has a non-empty open interior.
%%%%
Specifically, the interior  of $\overline{\stsp}(\hh)$ in $\mathcal{B}_{sa}^{1}(\hh)$ turns out to be
\be
\stsp(\hh)\,=\,\left\{\rho\in\stsp(\hh)\,|\;\;\rho>0\right\},
\ee
that is, quantum states in $\stsp(\hh)$ are invertible as linear operators.
%%%%
Since $\stsp(\hh)$ is open in $\mathcal{B}_{sa}^{1}(\hh)$ and the latter is an affine space, it follows that $\stsp(\hh)$ is a smooth manifold of real dimension $(n^{2}-1)$, where $\mathrm{dim}(\hh)=n$.
%%%%

\begin{definition}
The manifold $\stsp(\hh)$  is referred to as the space of faithful quantum states on $\hh$.
%%%%
\end{definition}

The manifold $\stsp(\hh)$ is the main character in quantum information geometry, but it is not the only one.
%%%%
Indeed, because of the Krein-Millman's theorem $\overline{\stsp}(\hh)$ is the closed convex hull of its set of extreme points.
%%%
These are precisely the rank-one projectors in $\bh$, which form a smooth manifold isomorphic with the complex projective space\footnote{Note that this also works in infinite-dimensions \cite{C-L-M-1983,C-M-P-1990}.} $\mathbb{CP}(\hh)$ associated with $\hh$.
%%%%
The diffeomorphism $F$ between  $\mathbb{CP}(\hh)$ and the space of rank-one projectors is given by
\be\label{eqn: pure states rank1 projections}
F([\psi]):=\frac{|\psi\rangle\langle\psi|}{\langle\psi|\psi\rangle}\equiv\rho_{\psi},
\ee
where $[\psi]$ is the equivalence class in  $\mathbb{CP}(\hh)$ representing $\psi\in\hh$.
%%%%%
In the following, we will denote the manifold of extreme points of $\overline{\stsp}(\hh)$ by $\stsp_{1}(\hh)$ where the subscript $1$ is telling us the rank of the density operators belonging to $\stsp_{1}(\hh)$.
%%%%

From a physical point of view, the extreme points in $\overline{\stsp}(\hh)$ are the \grit{pure states} of the theory.
%%%%
These are all those quantum states that can not be written as the convex combination of more than one quantum state.
%%%%%
All other states in $\overline{\stsp}(\hh)$, are called \grit{mixed states},  and are used to describe statistical mixtures of quantum states. The eigenvalue of the density matrix $\rho$ associated to an eigenstate $\ket{\psi}$ is then interpreted as the weight of the fraction in the mixture that lays in the state $\ket{\psi}$. 
%%%%
Thus, the fact that every $\rho$ is a positive semidefinite operator with unit trace implies that these weights are real, non-negative numbers whose sum is $1$, i.e., they form a probability vector.
%%%%%

The unitary group $\Uh$ naturally acts on $\hh$, and this action descends to the quotient $\mathbb{CP}(\hh)$ in the sense that we have
\be
\varphi_{U}([\psi]):=[\mathbf{U}\psi],
\ee
and, because of equation \eqref{eqn: pure states rank1 projections}, we immediately obtain the map
\be
\Phi_{U}(\rho_{\psi}):=\frac{|\mathbf{U}\psi\rangle\langle\mathbf{U}\psi|}{\langle\mathbf{U}\psi|\mathbf{U}\psi\rangle} = \mathbf{U}\rho_{\psi}\mathbf{U}^{\dagger}.
\ee
%%%%
This map is clearly well-defined on the whole space $\overline{\stsp}(\hh)$ of quantum states, and provides a left action of $\Uh$ on it given by
\be\label{eqn: unitary action}
\Phi_{U}(\rho ):=  \mathbf{U}\rho \mathbf{U}^{\dagger} \equiv \rho_{U} .
\ee
%%%%
The relevance of this action of $\Uh$ is partly due to the fact that the  time evolution of a closed quantum system is usually described in terms of it.
%%%%
Indeed, the time evolution of the wave function $\psi\in\mathcal{H}$ determined by Schr\"{o}dinger equation reads $U_{t}|\psi\rangle$, with $U_{t}=\mathrm{e}^{i tH}$, and with $H$ being the Hamiltonian operator of Schr\"{o}dinger equation. Thus, the dynamical evolution of a quantum state $\rho$,  determined by the one-parameter group of unitary operators associated with the Schr\"{o}dinger equation, reads $U_{t}\rho U_{t}^{\dagger}$.
%%%%

When $\mathrm{dim}(\hh)>2$, the manifolds $\stsp_{1}(\hh)$ and $\stsp(\hh)$ are not enough to fully describe the space of quantum states because the structure of the boundary of $\overline{\stsp}(\hh)$ becomes definitely more complex.
%%%
Indeed, $\overline{\stsp}(\hh)$ becomes the disjoint union of $n$ smooth manifolds of increasing dimensions.
%%%%
In particular,   the following result regarding the structure of $\overline{\stsp}(\hh)$ holds \cite{DA-F-2021, G-K-M-2005}.
%%%%%

\begin{theorem}\label{thm: stratification of quantum states}
The space of quantum states $\overline{\stsp}(\hh)$ of a finite-level quantum system with Hilbert space $\hh$ decomposes as the disjoint union
\be\label{eqn: decomposition of the space of states}
\overline{\stsp}(\hh)\,=\,\bigsqcup_{k=1}^{n}\,\stsp_{k}(\hh),
\ee
where each $\stsp_{k}(\hh)$ is made up of quantum states with fixed rank equal to $k$.
%%%
In particular, $\stsp_{n}(\hh)$ coincides with the space of faithful quantum states $\stsp(\hh)$ introduced before.
%%%
Moreover, $\stsp_{k}(\hh)$  is a smooth and connected real manifold of  dimension $2 n k - k^2 -1$.  
%%%%

\end{theorem}

It turns out that $\overline{\stsp}(\hh)$ with the decomposition in equation \eqref{eqn: decomposition of the space of states} is actually a stratified manifold \cite{DA-F-2021} whose strata are precisely  the manifolds $\stsp_{k}(\hh)$  with $k=1,..,n$.
%%%%
Accordingly, $\stsp_{1}(\hh)$ is often referred to as the minimal stratum, while $\stsp(\hh)=\stsp_{n}(\hh)$ is often referred to as the maximal stratum.
%%%%

\subsection{The nonlinear action of the general linear group}\label{subsec: nonlinear action of GL(H)}

It is important to note that the minimal stratum $\stsp_{1}(\hh)$ is not just a real smooth manifold, but it is a homogeneous space for the unitary group $\Uh$.
%%%%
Indeed, the action $(\mathbf{U},\rho_{\psi})\mapsto \mathbf{U}\rho_{\psi}\mathbf{U}^{\dagger}$ is clearly smooth and transitive on $\stsp_{1}(\hh)$.
%%%%
Moreover, this group action is strong enough to determine the so-called Fubini-Study metric on $\stsp_{1}(\hh)\cong\mathbb{CP}(\hh)$ as the unique (up to a constant factor) unitary invariant Riemannian metric tensor \cite{B-Z-2006}.
%%%%%
An interesting comparison between the Fubini-Study metric tensor and the Fisher-Rao metric tensor may be found in \cite{F-K-M-M-S-V-2010}.
%%%%%

As mentioned before, the map in equation \eqref{eqn: unitary action}  defines an action of $\Uh$ on the whole space of quantum states.
%%%%
Then, since the minimal stratum is a homogeneous space of $\Uh$, it is reasonable to investigate all other orbits of the action of $\Uh$ on the space of states.
%%%%
Clearly, since $\mathbf{U}$ is unitary, $\rho$ and $\mathbf{U}\rho\mathbf{U}^{\dagger}$ have the same eigenvalues and, conversely, if $\rho$ and $\rho'$ are quantum states having the same eigenvalues, then there is a unitary operator $\mathbf{U}$ such that $\rho'=\mathbf{U}\rho\mathbf{U}^{\dagger}$.
%%%%
It turns out that the sets of isospectral quantum states, i.e., quantum states having the same eigenvalues, are smooth homogeneous spaces of the unitary group.
%%%%
Clearly, if $\stsp_{k}(\hh)$ is a stratum in $\overline{\stsp}(\hh)$ with $k>1$, it is clear that  it contains an infinite number of manifolds of isospectral states of rank $k$.
%%%%
Therefore, the action of the unitary group is not enough to move transitively through the strata of the stratification of $\overline{\stsp}(\hh)$ given in equation \ref{eqn: decomposition of the space of states}.
%%%
It is a remarkable fact that we may ``enlarge'' the action of the unitary group on $\overline{\stsp}(\hh)$ to be an action of the general linear group $\mathcal{GL}(\hh)$ (which is the complexification of $\Uh$), in such a way that each $\stsp_{k}(\hh)$ becomes an homogeneous space of $\mathcal{GL}(\hh)$.
%%%%
Specifically, this action is given by
\be\label{eqn GL-action}
\rho\,\mapsto\,\alpha(\gr,\rho)\equiv\rho_{\gr}\,=\,\frac{\gr\,\rho\,\gr^{\dagger}}{\Tr_{\hh}(\gr\rho\gr^{\dagger})},
\ee
where $\gr\in\mathcal{GL}(\hh)$.
%%%
Notice that the factor in the  denominator is necessary to ensure that $\rho_{\gr}$ is a quantum state.
%%%
Quite interestingly, this action of $\mathcal{GL}(\hh)$ does not preserve the convex structure of $\overline{\stsp}(\hh)$, that is, it is a nonlinear action of the general linear group.
%%%%
However, it is readily seen that when $\gr$ is  a unitary operator, then we recover the isospectral action of $\Uh$ we discussed above.
%%%%%
The fact that the action $\alpha$ is transitive on each $\stsp_{k}(\hh)$ is proved in  \cite{C-I-J-M-2019}, and the idea behind it is to first use a unitary matrix to transform any quantum state $\rho'$ of rank $k$ into a quantum state $\rho'_{U}$ of rank $k$ commuting with a fixed quantum state $\rho$ of rank $k$, and then use a self-adjoint element in $\mathcal{GL}(\hh)$ commuting with $\rho$ to transform $\rho'_{U}$ in $\rho$.
%%%%%

Since $\stsp(\hh)$ is topologically trivial, it holds the global trivialization  $T\stsp(\hh)\cong \stsp(\hh)\times \mathcal{B}_{sa}^{0}(\hh)$ of the tangent bundle $T\stsp(\hh)$ of $\stsp(\hh)$, where 
\be
\mathcal{B}_{sa}^{0}(\hh)\,:=\,\left\{\mathbf{a}\in\bh \;|\;\;\;\mathbf{a}^{\dagger}=\mathbf{a},\;\;\Tr_{\hh}(\mathbf{a})=0\right\},
\ee
i.e., the space of traceless, self-adjoint operators.
%%%%
This global trivialization   depends on a particular identification of  the tangent space $T_{\rho}\stsp(\hh)$ with $\mathcal{B}_{sa}^{0}(\hh)$ which does not depend on the base point $\rho$.
%%%
Specifically, for every $\rho\in\stsp(\hh)$ and for a suitable $\epsilon>0$ (depending on $\rho$), we may consider the curve 
\be
\rho_{\mathbf{a}}(t):=\rho + t\mathbf{a}
\ee
with $\mathbf{a}\in\mathcal{B}_{sa}^{0}(\hh)$, which is a curve inside $\stsp(\hh)$ for all $|t|<\epsilon$, and thus obtain the identification of a tangent vector at $\rho$ with  $\mathbf{a}$.
%%%%
We call this identification of $T_{\rho}\stsp(\hh)$ with $\mathcal{B}_{sa}^{0}(\hh)$ the \grit{linear identification} because it makes use of the natural convex structure that $\overline{\stsp}(\hh)$ (and thus $\stsp(\hh)$) inherits from the ambient space.
%%%
However, note that the curve $\rho_{\mathbf{a}}(t)$ always escapes $\stsp(\hh)$, and also $\overline{\stsp}(\hh)$, after a finite time.
%%%%

The linear identification is probably the most used identification of tangent vectors at $\rho\in\stsp(\hh)$, however, it is definitely not unique.
%%%%
Indeed, since $\stsp(\hh)$ is a homogeneous space of $\mathcal{GL}(\hh)$, we may realize the tangent space $T_{\rho}\stsp(\hh)$ in terms of the fundamental vector fields of the action of $\mathcal{GL}(\hh)$ on $\stsp(\hh)$, and this will lead us to another identification of tangent vectors that, according to \cite{Ciaglia-2020,C-J-S-2020}, is particularly tailored for the monotone metric tensor known as the Bures-Helstrom metric tensor, as reviewed in subsection \ref{subsec: B-H metric tensor}.
%%%%
Specifically, we first note that every $\gr\in\mathcal{GL}(\hh)$ can be written as $\gr=\mathrm{e}^{\frac{1}{2}(\mathbf{a} + i\mathbf{b})}$, where $\mathbf{a},\mathbf{b}$ are self-adjoint elements in $\bh$ (recall that the Lie algebra of $\mathcal{GL}(\hh)$ is $\bh$).
%%%%
Then, we consider the curve $g(t) = e^{\frac{t}{2}(\mathbf{a} + i \mathbf{b})}$ starting at the identity element, and compute the associated fundamental tangent vector $\Gamma_{\mathbf{ab}}(\rho)$ at $\rho$ according to
\be\label{eqn: GL-fundamental tangent vectors}
\begin{split}
\Gamma_{\mathbf{ab}}(\rho) &= \frac{\mathrm{d}}{\mathrm{d}t} \left(\frac{g(t)\rho \, g^{\dagger}(t)}{\Tr_{\hh}(g(t)\rho \, g^{\dagger}(t))}  \right)_{t=0} = \\
&= \frac{1}{2}(\rho\mathbf{a} + \mathbf{a}\rho) - \Tr_{\hh}(\mathbf{a}\rho) \rho +  \frac{i}{2}( \mathbf{b} \rho - \rho \mathbf{b})\\ 
& =  \{\rho, \mathbf{a}\} - \Tr(\mathbf{a}\rho) \rho  + [[\rho,\mathbf{b}]]  
\end{split}
\ee
where we used equation \eqref{eqn: Jordan product} and $[[\rho,\mathbf{b}]]=\frac{i}{2}( \mathbf{b} \rho - \rho \mathbf{b})$.
%%%%%
Therefore, we obtain that every tangent vector $\mathbf{v}_{\rho}$ in $T_{\rho}\stsp(\hh)$ can be written as
\be\label{eqn: GL-fundamental tangent vectors 2}
\mathbf{v}_{\rho}= \{\rho, \mathbf{a}\} - \Tr(\mathbf{a}\rho) \rho  + [[\rho,\mathbf{b}]]= \{\rho,\mathbf{a} - \mathbb{I}_{\hh}\Tr_{\hh}(\rho\mathbf{a})\} + [[\rho,\mathbf{b}]],
\ee
where $\mathbb{I}_{\hh}$ is the identity operator on $\hh$.
%%%%
Clearly, when $\mathbf{a}=\mathbf{0}$ we obtain the action of the unitary group, and thus $\Gamma_{\mathbf{0b}}$ may be read as the fundamental vector field of the action of $\Uh$, and we set $\mathbb{X}_{\mathbf{b}}:=\Gamma_{\mathbf{0b}}$.
%%%%
Similarly, we write 
\be\label{eqn: gradient vector fields}
\mathbb{Y}_{\mathbf{a}}:=\Gamma_{\mathbf{a0}},
\ee
and we note that, unlike the $\mathbb{X}_{\mathbf{b}}$'s, these vector fields do not close a Lie algebra and thus do not give rise to a group action. 
%%%
It is then clear that every fundamental vector field $\Gamma_{\mathbf{ab}}$ of $\mathcal{GL}(\hh)$ on $\stsp(\hh)$ can be written as the sum
\be\label{eqn: hamiltonian + gradient vector fields}
\Gamma_{\mathbf{ab}}\,=\,\mathbb{Y}_{\mathbf{a}} + \mathbb{X}_{\mathbf{b}}.
\ee
%%%%%%%%
It turns out that, contrarily to the vector fields $\mathbb{X}_{\mathbf{b}}$, that generate the tangent space of the manifolds of isospectral states in $\stsp(\hh)$, the  $\mathbb{Y}_{\mathbf{a}}$'s are enough to generate the tangent space at each $\rho\in\stsp(\hh)$ \cite{C-J-S-2020}.
%%%%
Roughly speaking, this follows from the fact that the linear super-operator
\be
\mathbf{a}\,\mapsto A_{\rho}(\mathbf{a}):=\{\rho,\mathbf{a}\}
\ee
is invertible for all $\rho\in\stsp(\hh)$.
%%%%%
This means that every tangent vector $\mathbf{v}_{\rho}\in T_{\rho}\stsp(\hh)$ can be written as
\be\label{eqn: Jordan identification}
\mathbf{v}_{\rho} \,=\, \,\{\rho, \mathbf{a}\} - \Tr(\mathbf{a}\rho) \rho \,=\, \{\rho,\mathbf{a} - \mathbb{I}_{\hh}\Tr_{\hh}(\rho\mathbf{a})\}\equiv J_{\rho}^{\mathbf{a}}
\ee
for some self-adjoint $\mathbf{a}\in\bh$.
%%%%
We call this identification the \grit{Jordan identification} because it makes use of the Jordan product $\{\cdot,\cdot\}$ associated with the (scaled) anti-commutator.
%%%%%
Note that, in general, it may happen that the same $\mathbf{v}_{\rho}$ may be associated with different operators  $\mathbf{a}_{1}\neq\mathbf{a}_{2}$ through the Jordan identification.
%%%%

We will see that the vector fields $\mathbb{Y}_{\mathbf{a}}$   are the gradient vector fields, in the sense of Riemannian geometry, associated with linear functions by means of the so-called Bures-Helstrom monotone metric tensor (see also \cite{Ciaglia-2020,C-J-S-2020}).
%%%%

\begin{remark}
Since every $\stsp_{k}(\hh)$ is a smooth homogenous space of $\mathcal{GL}(\hh)$, equations \eqref{eqn: GL-fundamental tangent vectors} and \eqref{eqn: GL-fundamental tangent vectors 2} make sense for every $\rho\in\stsp_{k}(\hh)$ for all $k=1,..,n$ and give an identification of tangent vectors at $\rho\in\stsp_{k}(\hh)$ with elements in $\mathcal{B}_{sa}^{0}(\hh)$.
%%%
Clearly, we do not recover all $\mathcal{B}_{sa}^{0}(\hh)$ because $\rho$ is not full-rank. 
%%%%
However, it turns out that even if $k<n$, the vector fields $\mathbb{Y}_{\mathbf{a}}$  generate the tangent space at each $\rho\in\stsp_{k}(\hh)$ \cite{C-J-S-2020}, and this means that  the Jordan identification actually makes sense for all strata of the space of quantum states.
%%%%
Of course,  the proof can not rely on the invertibility of $A_{\rho}$ because the latter  is not invertible when $\rho$ is not full-rank.
%%%%

%%%%

\end{remark}

\subsection{The square-root identification}\label{eqn: square-root identification}

Once we have the transitive action of $\Glh$ on $\stsp(\hh)$  given in equation \eqref{eqn GL-action}, we can immediately define an infinite number of transitive actions of $\Glh$ on $\stsp(\hh)$ by suitably intertwining  $\alpha$ with a diffeomorphism, $\varphi$, of the cone of positive, invertible operators on $\hh$ in itself,  such that $\varphi(c\rho)=c\varphi(\rho)$ for all positive $c\in\mathbb{R}$.
%%%%
Specifically, it is a matter of direct inspection to check that 
\be
\alpha_{\varphi}(\gr,\rho):=\frac{\varphi^{-1}\left(\gr\,\varphi(\rho)\,\gr^{\dagger}\right)}{\Tr_{\hh}(\varphi^{-1}\left(\gr\varphi (\rho)\gr^{\dagger}\right))}
\ee
is again a transitive action of $\Glh$ on $\stsp(\hh)$.
%%%% 
In particular, if we take $\varphi(\rho)=\sqrt{\rho}$, we obtain the action $\alpha_{\varphi}\equiv \beta$ given by
\be
\beta(\gr,\rho):=\frac{\left(\gr\,\sqrt{\rho}\,\gr^{\dagger}\right)^{2}}{\Tr_{\hh}\left(\left(\gr\sqrt{\rho}\gr^{\dagger}\right)^{2}\right)}.
\ee
%%%%%
Then, proceeding as in equation \eqref{eqn: GL-fundamental tangent vectors}, we obtain the fundamental vector fields $\Upsilon_{\mathbf{ab}}$ given by
\be\label{eqn: sqrt-GL-fundamental tangent vectors}
\begin{split}
\Upsilon_{\mathbf{ab}}(\rho) &= \frac{\mathrm{d}}{\mathrm{d}t} \left(\beta(\gr(t),\rho)  \right)_{t=0}  =  \{\rho, \mathbf{a}\} + \sqrt{\rho}\,\mathbf{a}\sqrt{\rho} - 2\Tr(\mathbf{a}\rho) \rho  + [[\rho,\mathbf{b}]]  .
\end{split}
\ee
%%%%
Clearly, we can write
\be
\Upsilon_{\mathbf{ab}}\,=\,\mathbb{W}_{\mathbf{a}} + \mathbb{X}_{\mathbf{b}},
\ee
where $\mathbb{X}_{\mathbf{b}}=\Upsilon_{\mathbf{0}\mathbf{b}}=\Gamma_{\mathbf{0}\mathbf{b}}$  is a fundamental vector field for the action of the unitary group, and 
\be\label{eqn: square-root gradient vector fields}
\mathbb{W}_{\mathbf{a}}=\Upsilon_{\mathbf{a}\mathbf{0}},
\ee
is $\varphi$-related to the gradient vector field $\mathbb{Y}_{\mathbf{a}}=\Gamma_{\mathbf{a}\mathbf{0}}$.
%%%%   
Therefore, the vector fields $\mathbb{W}_{\mathbf{a}}$ are enough to generate the tangent space of $\stsp(\hh)$ at each $\rho$, and we call the identification of a tangent vector $\mathbf{v}_{\rho}$ given by
\be\label{eqn: sqrt-GL-fundamental tangent vectors 2}
\mathbf{v}_{\rho} = 2\{\sqrt{\rho},\{\sqrt{\rho}, \mathbf{a} - \mathbb{I}_{\hh}\Tr_{\hh}(\rho\mathbf{a})\}\} = \{\rho, \mathbf{a}\} + \sqrt{\rho}\,\mathbf{a}\sqrt{\rho} - 2\Tr(\mathbf{a}\rho) \rho  \equiv  S_{\rho}^{\mathbf{a}} ,
\ee
the \grit{square-root} identification.
%%%%
We will see that the vector fields $\mathbb{W}_{\mathbf{a}}$   are the gradient vector fields, in the sense of Riemannian geometry, associated with linear functions by means of the so-called Wigner-Yanase monotone metric tensor (see also \cite{Ciaglia-2020}).
%%%%

\subsection{The exponential identification}\label{subsec: exponential identification}

There is yet another quite interesting identification of tangent vectors that it is worth   recalling.
%%%%
To this, we first notice that the set $\pos(\hh)$ of positive, invertible operators on $\hh$ is an open subset of the vector space $\mathcal{B}_{sa}(\hh) $ of Hermitian (self-adjoint) linear operators on $\hh$.
%%%
Moreover, every $\mathbf{h}\in \mathcal{B}_{sa}(\hh)$ gives rise to an element in $\pos(\hh)$ by means of  $\mathrm{e}^{\mathbf{h}}$, and every $\rho\in\pos(\hh)$ gives rise to an element in $\mathcal{B}_{sa}(\hh)$ by means of $\ln(\rho)$.
%%%
Essentially, the map $\psi\colon\pos(\hh)\ra\mathcal{B}_{sa}(\hh)$ given by
\be
\psi(\rho)\,:=\,\ln(\rho)
\ee
is a diffeomorphism with inverse
\be
\psi^{-1}(\mathbf{h})\,=\,\mathrm{e}^{\mathbf{h}}\,.
\ee
%%%%
Inspired by what we have done in the previous subsection, we use $\psi$ to define an action of the Euclidean group on $\pos(\hh)$.
%%%%
Indeed, $\mathcal{B}_{sa}(\hh)$ is a real Euclidean space with respect to the (restriction of) the Hilbert-Schmidt product
\be
\langle\mathbf{h},\mathbf{k}\rangle\,=\,\mathrm{Tr}_{\hh}(\mathbf{h}\,\mathbf{k}),
\ee
and the Euclidean group acts on $\mathcal{B}_{sa}(\hh)$ as
\be\label{eqn: Eclidean group action}
A_{R,\mathbf{a}}(\mathbf{h})\,=\,R(\mathbf{h}) + \mathbf{a},
\ee
where $R$ is an element of the orthogonal group and $\mathbf{a}\in\mathcal{B}_{sa}(\hh)$.
%%%%
The unitary group $\Uh$ may be realised as a subgroup of the orthogonal group of $\mathcal{B}_{sa}(\hh)$ according to 
\be\label{eqn: coadjoint action of unitary group}
R_{\mathbf{U}}(\mathbf{h})\,:=\,\mathbf{U}\mathbf{h}\mathbf{U}^{\dagger},
\ee
because it is easily checked that $R_{\mathbf{U}}$ preserves the Euclidean product
\be
\langle R_{U}(\mathbf{h}),R_{U}(\mathbf{k})\rangle\,=\,\langle\mathbf{U}\mathbf{h}\mathbf{U}^{\dagger},\mathbf{U}\mathbf{k}\mathbf{U}^{\dagger}\rangle\,=\,\langle\mathbf{h},\mathbf{k}\rangle.
\ee 
%%%%
We thus obtain an action of the group $\Uh\rtimes_{R} \stav$  on $\mathcal{B}_{sa}(\hh)$ by  restricting the action of the Euclidean group given in equation \eqref{eqn: Eclidean group action}. This   can be transported to $\pos(\hh)$,  thus giving rise to 
\be
\widetilde{\Xi}((\mathbf{U},\mathbf{a}),\rho)=\mathrm{e}^{\mathbf{U}\ln(\rho)\mathbf{U}^{\dagger} + \mathbf{a}}.
\ee 
%%%%%%
We can now normalize $\widetilde{\Xi}$ to $\stsp(\hh)$ and obtain the action $\Xi$ of $\Uh\rtimes_{R} \stav$ on $\stsp(\hh)$ given by 
\be
\Xi((\mathbf{U},\mathbf{a}),\rho)\,:=\, \frac{\mathrm{e}^{\mathbf{U}\ln(\rho)\mathbf{U}^{\dagger} + \mathbf{a}}}{\Tr(\mathrm{e}^{\mathbf{U}\ln(\rho)\mathbf{U}^{\dagger} + \mathbf{a}})}.
\ee

\begin{remark}
The Lie group  $\Uh\rtimes_{R} \stav$  is diffeomorphic to the cotangent bundle of the unitary group.
%%%
Indeed, if $G$ is any Lie group, the cotangent space $T^{*}G\cong G\times\mathfrak{g}^{*}$  is endowed with the structure of  Lie group \cite{A-G-M-M-1994,A-G-M-M-1998} according to
\be
(\gr_{1},\,a_{1})\cdot(\gr_{2},\,a_{2})\,:=\,(\gr_{1}\gr_{2},\,Ad_{\gr_{1} }^{*}(a_{2}) + a_{1}),
\ee
where $Ad^{*}$ is the dual of the adjoint action of $G$ on its Lie algebra $\mathfrak{g}$.
%%%
The resulting Lie group is also denoted by $G\rtimes_{Ad^{*}}\mathfrak{g}^{*}$  to emphasise the fact that the group structure is associated with a semidirect product.
%%%
Now, when $G=\Uh$, its Lie algebra $\mathfrak{g}$ is given by skew-adjoint operators on $\hh$ according to
\be
i\mathbf{b}\,\mapsto\,\mathbf{U}=\mathrm{e}^{i\mathbf{b}},
\ee
where $\mathbf{b}$ is an Hermitian operator.
%%%
Then, we can identify the dual space $\mathfrak{g}^{*}$ with the vector space $\stav$ of Hermitian operators by means of  the pairing
\be
\langle\mathbf{a},i\mathbf{b}\rangle\,:=\,\mathrm{Tr}(\mathbf{a}\,\mathbf{b}).
\ee
%%%
Consequently, the coadjoint action reads
\be
Ad_{\mathbf{U} }^{*}(\mathbf{a})\,=\,\mathbf{U}\,\mathbf{a}\,\mathbf{U}^{\dagger}\,=\, R_{\mathbf{U}}(\mathbf{a}),
\ee
where we used  equation \eqref{eqn: coadjoint action of unitary group} in the last equality, and  we conclude that $\Uh\rtimes_{R} \stav$  is actually diffeomorphic to the Lie group $T^{*}\Uh\equiv\Uh\rtimes_{Ad^{*}}\mathfrak{g}^{*}$ as claimed.

\end{remark}

The fundamental vector fields  $\Psi_{\mathbf{ab}} $ of $\Xi $ are easily computed to be 
\be
\Upsilon_{\mathbf{ab}} \,=\,\mathbb{Z}_{\mathbf{a}}  + \mathbb{X}_{\mathbf{b}}
\ee
where the $\mathbb{X}_{\mathbf{b}}$'s are the vector fields generating the standard action of the unitary group, and the $\mathbb{Z}_{\mathbf{a}} $'s are given by
\be\label{eqn:  exponential identification 2}
\mathbb{Z}_{\mathbf{a}} (\rho)\,=\,\frac{\mathrm{d}}{\mathrm{d}t}\left(\frac{\mathrm{e}^{ \ln(\rho)  + t\mathbf{a}}}{\Tr\left(\mathrm{e}^{ \ln(\rho)  + t\mathbf{a}}\right)}\right)_{t=0}=\int_{0}^{1}\,\mathrm{d}\lambda\,\left(\rho^{\lambda}\,\mathbf{a} \,\rho^{1-\lambda}\right) - \Tr(\rho\,\mathbf{a})\,\rho\,.
\ee
%%%%
By construction, the $\mathbb{Z}_{\mathbf{a}}$'s are enough to generate the tangent space at each $\rho\in\stsp(\hh)$,  and we call the identification of a tangent vector $\mathbf{v}_{\rho}$ given by
\be\label{eqn:  exponential identification 3}
\mathbf{v}_{\rho}=\int_{0}^{1}\,\mathrm{d}\lambda\,\left(\rho^{\lambda}\,\mathbf{a} \,\rho^{1-\lambda}\right) - \Tr(\rho\,\mathbf{a})\,\rho \equiv E_{\rho}^{\mathbf{a}}  ,
\ee
the \grit{exponential} identification.
%%%%
We will see that the vector fields  $\mathbb{Z}_{\mathbf{a}}$ are the gradient vector fields, in the sense of Riemannian geometry, associated with linear functions by means of the so-called Bogoliubov-Kubo-Mori monotone metric tensor (see also \cite{Ciaglia-2020}).
%%%%
%%%%
It is also worth mentioning the recent work \cite{A-L-2021} where the finite transformations associated with the vector fields $\mathbb{Z}_{\mathbf{a}}$ are exploited in the definition of a Hilbert space structure on $\stsp(\hh)$.

\section{Riemannian geometries of  quantum states}\label{sec: Cencov and Petz}

In the finite dimensional setting, the basic building blocks of Classical Information Geometry are probability distributions on a finite sample space. Let us denote this sample space with $\mathcal{X}_n$, where the subscript $n$ indicates its cardinality. As it is well-known, the space $\mathcal{P}(\mathcal{X}_{n})$ of probability distributions on $\mathcal{X}_{n}$ can be identified with the unit simplex in the Euclidean space $\mathbb{R}^{n}$, that is, with the set
\be\label{eqn: the simplex}
\overline{\Delta}_{n}:=\left\{ \mathbf{p}\in\mathbb{R}^{n}\;\colon \; p^{j}\geq 0, \; \sum_{j=1}^{n}\,p^{j}=1\right\}.
\ee
%%%%
This closed convex set is the closure of the convex set $\Delta_{n}$ of strictly positive probability vectors, that is, elements in $\overline{\Delta}_{n}$ for which $p^{j}>0$.
%%%%
The set  $\Delta_{n}$ is a smooth manifold (actually, it is an open submanifold of the hyperplane in $\mathbb{R}^{n}$ determined by the condition $ \sum_{j=1}^{n} x^j=1$), and it is referred to as the open interior of the n-simplex.
%%%%%
Loosely speaking, the n-simplex $\overline{\Delta}_{n}$ may be thought of as the classical counterpart of the space $\overline{\stsp(\hh)}$ of quantum states on an n-dimensional Hilbert space $\hh$, while the open interior $\Delta_{n}$ may be thought of as the classical counterpart of the manifold $\stsp(\hh)$ of invertible quantum states.
%%%%%

The manifold $\Delta_{n}$ can be endowed with a particular Riemannian metric tensor $\GG_{FR}^{n}$, known as the Fisher-Rao metric tensor, whose explicit expression in the over-complete, Cartesian coordinate chart\footnote{See section \ref{sec: unfolding of quantum states} for a thorough discussion on this.} $\{p^{j}\}_{j=1,...,n}$ inherited by $\mathbb{R}^{n}$ reads
\be\label{eqn: Fisher-Rao metric tensor}
\GG_{FR}^{n}=\sum_{j=1}^{n}\frac{1}{p^{j}}\,\mathrm{d}p^{j}\otimes\mathrm{d}p^{j}.
\ee
%%%%%
The latter was introduced by Rao \cite{Rao-1945} following Fisher \cite{Fisher-1922} and  Mahalanobis  \cite{Mahalanobis-1936}.
%%%%%
It is really impossible to overestimate the significance of $\GG_{FR}^{n}$ in Classical Information Geometry, Statistics, and related fields. Much of this importance has to be ascribed to a peculiar feature of $\GG_{FR}^{n}$ that has been elucidated by Cencov \cite{Cencov-1982} in the finite-dimensional case, and by others in different infinite-dimensional settings \cite{A-J-L-S-2017,B-B-M-2016}.
%%%%%
Remaining in the finite-dimensional case, it turns out that $\GG_{FR}^{n}$ is the unique Riemannian metric tensor on $\Delta_{n}$ (up to an overall multiplicative factor) whose associated distance function is invariant under \grit{congruent embeddings}.
%%%%
Specifically, given $\Delta_{n}\subset\mathbb{R}^{n}$ and $\Delta_{m}\subset\mathbb{R}^{m}$,  we first define a Markov kernel $M$ between $\Delta_{n}$ and $\Delta_{m}$ to be a stochastic linear map $M\colon\mathbb{R}^{n}\ra\mathbb{R}^{m}$.
%%%%%
Then, we define $M$ to be a \grit{congruent embedding} if it is also such that the image $M(\Delta_{n})\subseteq\Delta_{m}$ is diffeomorphic to $\Delta_{n}$.
%%%%
In particular, this requires that $n\geq m$.
%%%%
A typical example of a congruent embedding is given by a permutation in $\mathbb{R}^{n}$.
%%%%%
Then, Cencov's uniqueness result states that $\GG_{FR}^{n}$ is the only Riemannian metric tensor on $\Delta_{n}$ (up to an overall multiplicative factor) satisfying
\be
M^{*}\GG_{FR}^{m}=\GG_{FR}^{n}
\ee
for every congruent embedding $M$.
%%%%
It is clear that this result is not properly a result on  $\GG_{FR}^{n}$ for a fixed $n\in\mathbb{N}$, but, rather, it refers to a family of Riemannian metric tensors defined on the family of finite-dimensional simplexes.
%%%%
Indeed, Cencov's original formulation makes use of the power of Category Theory to correctly handle this instance.
%%%%%%

\vsp

Since we are interested in Quantum Information Geometry, it is reasonable to ask if there is an analogue of Cencov's theorem in this context.
%%%%%%
Once properly reformulated, it turns out that Cencov's remarkable uniqueness result does not apply to the quantum case, and we get an infinite number of quantum counterparts of the Fisher-Rao metric tensors, first discovered by Cencov and Morozowa \cite{C-M-1991}, and then completely classified by Petz \cite{Petz-1996}.
%%%%
We will now briefly review these results.
%%%%

In the context of quantum information theory,  of particular importance are the linear maps $\Phi\colon\mathcal{B}(\hh)\ra\mathcal{B}(\mathcal{K})$ that are completely-positive and trace preserving (CPTP), and such that $\Phi(\stsp(\hh ))\subseteq\stsp(\mathcal{K})$.
%%%%
We recall that a linear map $\Phi\colon\mathcal{B}(\hh)\ra\mathcal{B}(\mathcal{K})$ is a CPTP map if it maps positive operators to positive operators, and if it is such that the map $\Phi\otimes\mathrm{Id}_{n}\colon\bh\otimes\mathcal{M}_{n}(\mathbb{C})\ra\mathcal{B}(\mathcal{K})\otimes\mathcal{M}_{n}(\mathbb{C})$ given by  
\be
\Phi\otimes \mathrm{Id}_{n}(\mathbf{a}\otimes M)\,:=\,\Phi(\mathbf{a})\otimes M
\ee
sends positive elements to positive elements for all $n>0$ \cite{Choi-1975, Stinespring-1955}.
%%%%
Inspired by \cite{Cencov-1982,C-M-1991}, a CPTP map satisfying the additional requirement $\Phi(\stsp(\hh ))\subseteq\stsp(\mathcal{K})$ is called a \grit{quantum Markov map}.
%%%%
These  maps are at the heart of Petz's classification of the family of monotone metric tensors we will review below, and  may be considered to be the quantum counterpart (in the finite-dimensional, Hilbert space framework) of the congruent embeddings appearing in Cencov's theorem on the uniqueness of the Fisher-Rao metric tensor in the classical case.
%%%%%
The family of quantum Markov maps is a distinguished family of transformations of quantum states that turns out to be physically relevant, either because of statistical considerations \cite{Holevo-2001,Holevo-2011}, or because of dynamical considerations \cite{G-K-S-1976,Lindblad-1976}.
%%%%

The problem raised by Cencov and Morozowa \cite{C-M-1991}, and then completely solved by Petz \cite{Petz-1996}, was to find all the Riemannian metric tensors $\GG^{\hh}$ on $\stsp(\hh)$ such that
\be\label{eqn: monotonicity condition for quantum metrics}
\Phi^{*}\GG^{\mathcal{K}} \leq \GG^{\hh}
\ee
for all quantum Markov maps.
%%%%
Note that this quantum problem does not require equality  of the metric tensor as in the classical case, but only the fact that $\GG^{\hh} - \Phi^{*}\GG^{\mathcal{K}}$ is positive semidefinite.
%%%%
This instance may be ascribed to the fact that the image of $\stsp(\hh)$ through the quantum Markov map $\Phi$ is not required to be diffeomorphic to $\stsp(\hh)$ as it happened for simplexes and congruent embeddings.
%%%%
However, it is clear that, when $\Phi(\stsp(\hh))$ is actually diffeomorphic with $\stsp(\hh)$, as it may happen when $\Phi$ is induced by unitary operators through $\Phi(\rho)=\mathbf{U}\rho\mathbf{U}^{\dagger}$, then equation \eqref{eqn: monotonicity condition for quantum metrics} must hold with an equality.
%%%%
Consequently, all the Riemannian metric tensors satisfying equation \eqref{eqn: monotonicity condition for quantum metrics} will be unitary invariant, and all of them will be invariant with respect to the quantum analogue of a congruent embedding.
%%%%

\begin{remark}

Like in the classical case, the problem is not properly related to a single Riemannian metric tensor $\GG^{\hh}$ defined on a manifold of states $\stsp(\hh)$  for a fixed system with Hilbert space $\hh$, but, rather, it  refers to a family of Riemannian metric tensors defined on the family of manifolds of faithful quantum states of all possible finite-dimensional quantum system.
%%%%%
Unlike the classical case, there is not yet a proper categorical formulation of the results presented in \cite{C-M-1991,Petz-1996}.
%%%%%
It would be interesting to develop a unification of these two categories into a single category that would allow for a better understanding of those properties that are exquisitely quantum.
%%%%
Specifically, one may define the category $\mathsf{FdFvNS}$ whose objects are spaces of faithful states on finite-dimensional von-Neumann algebras, and whose morphisms are maps that are dual to completely-positive, unital maps between von-Neumann algebras.
%%%%
Indeed, on the one hand,  probability vectors may be identified with states on finite-dimensional, Abelian von Neumann algebras, and, through this identification, every congruent embedding is dual to a completely-positive, unital map  between suitable Abelian von-Neumann algebras.
%%%
On the other hand, quantum states on $\hh$ are precisely states (in the functional analytic sense) on the finite-dimensional von Neumann algebra $\bh$, and quantum Markov maps are, essentially by definition, dual to  completely-positive, unital maps.
%%%%%
As it is, the category $\mathsf{FdFvNS}$ would take into account only basic properties of classical and quantum systems, but would be nevertheless useful in obtaining a unification of the classification of metric tensors performed by Cencov in the classical case and by Petz in the quantum case.
%%%
Of course, to deal with more advanced aspects like composition of systems, marginalization, and similar instances, one would be forced to consider more complex categories along the lines of the classical and quantum Markov categories defined in \cite{Fritz-2020,Parzygnat-2020}.
%%%%

\end{remark}

A Riemannian metric tensor $\GG^{\hh}$ on $\stsp(\hh)$ consists of an assignment, for every $\rho\in\stsp(\hh)$, of a positive definite, symmetric, bilinear form $\GG_{\rho}^{\hh}$ on $T_{\rho}\stsp_{n}$ which is  smooth in the sense that the map
\be\label{eqn: metric tensor 1}
\rho\,\mapsto\,\GG_{\rho}^{\hh}(V_{\rho},W_{\rho})
\ee
is a smooth map for all vector fields $V,W$ on $\stsp(\hh)$.
%%%%
If we exploit the trivialization $T\stsp(\hh)\cong\stsp(\hh)\times \mathcal{B}_{sa}^{0}(\hh)$ induced by the linear identification, $V_{\rho}$ and $W_{\rho}$ can be identified with elements in $\mathcal{B}_{sa}^{0}(\hh)$, and the smoothness assumption on $\GG^{\hh}$ translates into the smothness of the function
\be\label{eqn: metric tensor 2}
\rho\,\mapsto\,\GG_{\rho}^{\hh}(\mathbf{a},\mathbf{b})
\ee
for all $\mathbf{a},\mathbf{b}\in\mathcal{B}_{sa}^{0}(\hh)$.
%%%%
This means that $\GG_{\rho}^{\hh}$ can be thought of as a symmetric bilinear form on $\mathcal{B}_{sa}^{0}(\hh)$, and every such bilinear form is obtained by restriction of a Hermitian form on $\bh$.
%%%%
Moreover, every Hermitian form $H^{\hh}$ on $\bh$ can be written in terms of the standard Hilbert-Schmidt product $\langle,\rangle_{\hh}^{HS}$ on $\bh$ defined by
\be
\langle\mathbf{a},\mathbf{b}\rangle_{\hh}^{HS}:=\Tr_{\hh}\left(\mathbf{a}^{\dagger}\,\mathbf{b}\right)
\ee
for all $\mathbf{a},\mathbf{b}\in\bh$.
%%%%
Specifically, 
\be
H^{\hh}(\mathbf{a},\mathbf{b})=\langle\mathbf{a},T^{\hh}(\mathbf{b})\rangle_{\hh}^{HS}=\Tr_{\hh}\left(\mathbf{a}^{\dagger}\,T^{\hh}(\mathbf{b})\right),
\ee
where $T^{\hh}$ is a positive definite  (with respect to the Hilbert-Schmidt product)  operator on $\bh$.
%%%%
Therefore, we have that $\GG_{\rho}^{\hh}$ can always be written as the restriction to $\mathcal{B}_{sa}^{0}(\hh)$ of an Hermitean product on $\bh$ of the form
\be\label{eqn: Hermitean form}
H_{\rho}^{\hh}(\mathbf{a},\mathbf{b})=\langle\mathbf{a},T_{\rho}^{\hh}(\mathbf{b})\rangle_{\hh}^{HS}=\Tr_{\hh}\left(\mathbf{a}^{\dagger}\,T_{\rho}^{\hh}(\mathbf{b})\right),
\ee
where the positive definite operator $T_{\rho}^{\hh}$ depends on $\rho$.
%%%
Of course, we have to ensure  that the map $\rho\mapsto T_{\rho}^{\hh}$ is such that   $\rho\mapsto H_{\rho}^{\hh}(\mathbf{a},\mathbf{b})$ is smooth for all $\mathbf{a},\mathbf{b}\in\mathcal{B}_{sa}^{0}(\hh)$. 
%%%%

Petz's classification states that a Riemannian metric tensors on $\stsp(\hh)$ satisfies equation \eqref{eqn: monotonicity condition for quantum metrics} if and only if its associated superoperator is of the form   $T_{\rho}^{f}= (K^{f}_{\rho})^{-1}$, with $f\colon(0,\infty)\ra(0,\infty)$  an operator monotone function such that
\be\label{eqn: Petz function}
f(x)=xf(x^{-1}),\quad f(1)=1,
\ee
and 
\be\label{eqn: Petz metric 2}
K^{f}_{\rho}= f(L_{\rho}\,R_{\rho^{-1}})\,R_{ \rho },
\ee
where $L_{\rho}(\mathbf{a})=\rho\mathbf{a}$ and $R_{\rho}(\mathbf{a})=\mathbf{a}\rho$.
%%%%
This means that every Riemannian metric tensors on $\stsp(\hh)$ satisfying equation \eqref{eqn: monotonicity condition for quantum metrics} can be written as $\GG^{\hh}_{f}$ with 
\be\label{eqn: Petz metric 1}
(\GG^{\hh}_{f})_{\rho}(\mathbf{v}_{\rho},\mathbf{w}_{\rho})=\langle\mathbf{v}_{\rho}, T^{f}_{\rho}(\mathbf{w}_{\rho})\rangle^{HS}_{\hh}.
\ee
%%%%
If we introduce the operators $\mathbf{e}^{\rho}_{lm}$ diagonalizing $\rho$, that is, such that
\be
\rho=\sum_{j=1}^{n}\,p_{j}^{\rho}\,\mathbf{e}_{jj}^{\rho},
\ee
we can also introduce the super-operators  $E_{kj}^{\rho}$ acting on $\bh$ according to
\be\label{eqn: rho super eigenprojectors}
E_{kj}^{\rho}\left( \mathbf{e}^{\rho}_{lm}\right)\,=\,\delta_{jl}\,\delta_{km}\mathbf{e}_{jk}^{\rho},
\ee
and it is then a matter of straightforward computation to check that
\be
K^{f}_{\rho}=\sum_{j,k=1}^{n}\,p_{k}^{\rho}\,f\left(\frac{p_{j}^{\rho}}{p_{k}^{\rho}}\right)\,E_{kj}^{\rho}
\ee
where $p_{1}^{\rho},...,p_{n}^{\rho}$ are the eigenvalues of $\rho$.
%%%
Therefore, we also have
\be\label{eqn: eigendecomposition of Petz superoperator}
T^{f}_{\rho}=\sum_{j,k=1}^{n}\,\left(p_{k}^{\rho}\,f\left(\frac{p_{j}^{\rho}}{p_{k}^{\rho}}\right)\right)^{-1}\,E_{kj}^{\rho}.
\ee
%%%% 
Now, whenever $[\mathbf{w}_{\rho},\rho]=0$, using equation \eqref{eqn: eigendecomposition of Petz superoperator}  it is easily seen that 
\be
(\GG^{\hh}_{f})_{\rho}(\mathbf{v}_{\rho},\mathbf{w}_{\rho})=\sum_{j=1}^{n}\frac{v_{\rho}^{jj} w_{\rho}^{jj}}{p_{j}^{\rho}},
\ee
where $v_{\rho}^{jj}$ and $w_{\rho}^{jj}$ are the diagonal elements of $\mathbf{v}_{\rho}$ and $\mathbf{w}_{\rho}$ with respect to the basis of eigenvectors of $\rho$.
%%%%
It is relevant to note then that in this case we have
\be\label{eqn: Petz and F-R}
(\GG^{\hh}_{f})_{\rho}(\mathbf{v}_{\rho},\mathbf{w}_{\rho})\,=\,(\GG_{FR})_{\vec{p}}\,(\vec{a},\vec{b}),
\ee
where $\GG_{FR}$ is the classical Fisher-Rao metric tensor on probability distributions, and we have set $\vec{p}=(p_{1}^{\rho},...,p_{n}^{\rho})$, $\vec{a}=(v_{\rho}^{11},...,v_{\rho}^{nn})$, and $\vec{b}=(w_{\rho}^{11},...,w_{\rho}^{nn})$.
%%%%
Equation \eqref{eqn: Petz and F-R} holds for every choice of the operator monotone function $f$.
%%%%%

\vsp

In the rest of this section, we will briefly discuss three distinguished members of the family of quantum monotone Riemannian metric tensors classified by Petz.
%%%%%
Specifically, the so-called Bures-Helstrom metric tensor, the Wigner-Yanase metric tensor, and the Bogoliubov-Kubo-Mori metric tensor.
%%%%
These members stand out for different reasons.
%%%%
%%%
In particular, we will see that the Bures-Helstrom metric is intimately connected with the Jordan identification and the action of $\mathcal{GL}(\hh)$ introduced in subsection \ref{subsec: nonlinear action of GL(H)}, that the Wigner-Yanase metric is intimately connected with the the square-root identification and the action of $\mathcal{GL}(\hh)$ introduced in subsection \ref{eqn: square-root identification}, whereas the Bogoliubov-Kubo-Mori metric tensor is intimately connected with the exponential identification and the action of $T^{*}\SUh$ introduced in subsection \ref{subsec: exponential identification}.
%%%%% 

\subsection{The Bures-Helstrom metric tensor}\label{subsec: B-H metric tensor}

An example of   monotone quantum Riemannian metric tensor belonging to the family classified by Petz is the so-called Bures-Helstrom metric tensor $\GG_{BH}^{\hh}$ \cite{B-Z-2006, Dittmann-1993, Dittmann-1995, Helstrom-1967, Helstrom-1968, Helstrom-1969, Helstrom-1976, S-A-G-P-2020, Uhlmann-1992,  Safranek-2017, Safranek-2018}.
%%%%
This Riemannian metric corresponds to the choice $f(x)=\frac{1 +x}{2}$, from which it follows that equation \eqref{eqn: Petz metric 2} becomes
\be
f(L_{\rho}\,R_{\rho^{-1}})\,R_{\rho}=\frac{1}{2}\left(L_{\rho} + R_{\rho}\right)\equiv A_{\rho},
\ee
and thus, writing $\GG^{\hh}_{f}\equiv\GG_{BH}^{\hh}$, equation \eqref{eqn: Petz metric 1} leads us to
\be\label{eqn: BH metric tensor}
(\GG^{\hh}_{BH})_{\rho}(\mathbf{v}_{\rho},\mathbf{w}_{\rho})\,=\,\Tr_{\hh}\left(\mathbf{v}_{\rho}\,A^{-1}_{\rho}(\mathbf{w}_{\rho})\right)\,,
\ee
where $\mathbf{v}_{\rho},\mathbf{w}_{\rho}\in\mathcal{B}_{sa}^{0}(\hh)\cong T_{\rho}\stsp(\hh)$.
%%%%%
%%%%%
Equation \eqref{eqn: BH metric tensor} implies that, when we employ the linear identification of tangent vectors at $\rho\in\stsp$, the explicit computation of the Bures-Helstrom scalar product requires the computation of the inverse of the operator $A_{\rho}$ at each $\rho$.
%%%%
This task may be computationally demanding.
%%%%
However, the linear identification of tangent vectors is not the only possible, and indeed, by exploiting the Jordan identification introduced in equation \eqref{eqn: Jordan identification}, we immediately conclude that 
\be\label{eqn: BH metric tensor 2}
(\GG^{\hh}_{BH})_{\rho}(J_{\rho}^{\mathbf{a}},J_{\rho}^{\mathbf{b}})\,=\,\Tr_{\hh}\left(\rho\,\left\{ \mathbf{a} , \mathbf{b} \right\}\right) - \Tr_{\hh}(\rho\, \mathbf{a} )\,\Tr_{\hh}(\rho\mathbf{b} ).
\ee
%%%%
This expression for the Bures-Helstrom metric tensor is clearly easier to handle than the one in equation \eqref{eqn: BH metric tensor}, but it  forces  to work with the lesser known  Jordan identification for tangent vectors.
%%%%%

The expression of the Bures-Helstrom metric tensor in terms of the Jordan identification also leads us to appreciate an unexpected link between this metric tensor and the nonlinear action of $\mathcal{GL}(\hh)$ introduced in subsection \ref{subsec: nonlinear action of GL(H)}.
%%%
Recalling the definition of the gradient vector fields $\mathbb{Y}_{\mathbf{a}},\mathbb{Y}_{\mathbf{b}}$ in equations  \eqref{eqn: gradient vector fields} and \eqref{eqn: GL-fundamental tangent vectors}, we see that
\be
\left(\GG_{BH}^{\hh}\right)_{\rho}\left(\mathbb{Y}_{\mathbf{a}}(\rho),\mathbb{Y}_{\mathbf{b}}(\rho)\right)=(\GG^{\hh}_{BH})_{\rho}(J_{\rho}^{\mathbf{a}},J_{\rho}^{\mathbf{b}})=\Tr_{\hh}\left(\rho\,\{\mathbf{a},\mathbf{b}\}\right) - \Tr_{\hh}(\rho\mathbf{a})\,\Tr_{\hh}(\rho\mathbf{b}),
\ee
and thus
\be
\GG^{\hh}_{BH}\left(\mathbb{Y}_{\mathbf{a}},\mathbb{Y}_{\mathbf{b}}\right)=\mathbb{Y}_{\mathbf{b}}(l_{\mathbf{a}}),
\ee
where $l_{\mathbf{a}}$ is the linear function $l_{\mathbf{a}}(\rho)=\Tr_{\hh}(\rho\mathbf{a})$.
%%%%
Since the gradient vector fields generate the tangent space at each $\rho\in\stsp(\hh)$, and since the differential of linear functions on $\stsp(\hh)$ generate the cotangent space at each point $\rho\in\stsp(\hh)$, we conclude\footnote{See also \cite{Ciaglia-2020, C-J-S-2020} for alternative proofs of this instance.} that the gradient vector field  (in the sense of Riemannian geometry) associated with the linear function  $l_{\mathbf{a}}$ by means of the Bures-Helstrom metric tensor $\GG^{\hh}_{BH}$  is precisely the vector field  $\mathbb{Y}_{\mathbf{a}}$.
%%%
This justifies the name gradient vector field we chose for  $\mathbb{Y}_{\mathbf{a}}$ when we introduced it.
%%%%%

The Bures-Helstrom metric tensor is particularly relevant for all the information-theoretical tasks related with the quantum formulation of estimation theory \cite{C-J-S-2020-02, L-Y-L-W-2020, Paris-2009, Suzuki-2019, S-Y-H-2020, T-A-D-2020} because it allows to give the lowest quantum version of the classical Cramer-Rao bound for unbiased estimators \cite{Fuchs-1996}.
%%%
From a more geometrical point of view,  $\GG_{BH}^{\hh}$ is naturally connected with the concept of purification for quantum states \cite{D-U-1999, Uhlmann-1986, Uhlmann-1992}, and with the Jordan product (anticommutator) among self-adjoint operators \cite{C-J-S-2020}.
%%%%
Moreover,  if we consider the GNS Hilbert space $\hh_{mm}$ associated with the maximally mixed state on $\hh$, it turns out there is a Riemannian submersion from the unit sphere in $\hh_{mm}$ endowed with the standard round metric and the manifold $\stsp(\hh)$ endowed with the Bures-Helstrom metric tensor \cite{C-J-S-2020}.
%%%%
Since the geometrical properties of the n-sphere with the round metric are well-known, with the help of the Riemannian submersion discussed above it is possible to compute the geodesics, the Riemann curvature, the Ricci tensor, the scalar curvature, and the sectional curvature for $(\stsp(\hh),\GG_{BH}^{\hh})$, and thus understand  all the geometrical features of this Riemannian manifold.
%%%%%

\subsection{The Wigner-Yanase metric tensor }\label{subsec: W-Y metric tensor}

Another example of monotone quantum Riemannian metric tensor belonging to the family classified by Petz  is the so-called Wigner-Yanase metric tensr $\GG_{WY}^{\hh}$ \cite{G-I-2001,G-I-2003,Hasegawa-1993,Hasegawa-1995,Hasegawa-2003,H-P-1997,Jencova-2003-2}.
%%%%%
This Riemannian metric corresponds to the choice $f(x)=(1 + \sqrt{x})^{2}$, from which it follows that equation \eqref{eqn: Petz metric 2} becomes
\be
f(L_{\rho}\,R_{\rho^{-1}})\,R_{\rho}= \left(L_{\sqrt{\rho}} + R_{\sqrt{\rho}}\right)\equiv 4 A_{\sqrt{\rho}}^{2},
\ee
and thus, by writing $\GG_{f}^{\hh}=\GG_{WY}^{\hh}$, equation \eqref{eqn: Petz metric 1} leads  to
\be\label{eqn: WY metric tensor}
(\GG^{\hh}_{WY})_{\rho}(\mathbf{v}_{\rho},\mathbf{w}_{\rho})\,=\,\frac{1}{4}\,\Tr_{\hh}\left(\mathbf{v}_{\rho}\,A^{-2}_{\sqrt{\rho}}(\mathbf{w}_{\rho})\right)\,,
\ee
where $\mathbf{v}_{\rho},\mathbf{w}_{\rho}\in\mathcal{B}_{sa}^{0}(\hh)\cong T_{\rho}\stsp(\hh)$.
%%%%%

Analogously to what happens for the Bures-Helstrom metric tensor, using the linear identification for tangent vectors forces us to compute first the square-root of $\rho$ and then the inverse of the operator $A^{-2}_{\sqrt{\rho}}$ in order to explicitly compute the Wigner-Yanase scalar product between $\mathbf{v}_{\rho}$ and $\mathbf{w}_{\rho}$.
%%%%%
However, if we exploit the square-root identification introduced in equation \eqref{eqn: sqrt-GL-fundamental tangent vectors 2}, we obtain
\be\label{eqn: WY metric tensor 2}
(\GG^{\hh}_{WY})_{\rho}(S_{\rho}^{\mathbf{a}},S_{\rho}^{\mathbf{b}})\,=\,\Tr_{\hh}\left(\rho\,\{\mathbf{a},\mathbf{b}\}\right) + \Tr_{\hh}\left(\sqrt{\rho}\,\mathbf{a}\sqrt{\rho}\mathbf{b}\right) - 2\Tr_{\hh}(\rho\mathbf{a})\,\Tr_{\hh}(\rho\mathbf{b})\,.
\ee
%%%%
Again, this expression is easier to handle than the one in equation \eqref{eqn: WY metric tensor}, but we have to deal with the square-root identification for tangent vectors which is less used and thus one may feel less comfortable to work with.
%%%%

Proceeding in analogy with what was done for $\GG_{BH}^{\hh}$, we recall  the expression of the vector fields  $\mathbb{W}_{\mathbf{a}}$ and $\mathbb{W}_{\mathbf{b}}$ given by equation  \eqref{eqn: square-root gradient vector fields}    and equation \eqref{eqn: sqrt-GL-fundamental tangent vectors}, and we immediately conclude that
\be
\begin{split}
\left(\GG_{WY}\right)_{\rho}\left(\mathbb{W}_{\mathbf{a}}(\rho),\mathbb{W}_{\mathbf{b}}(\rho)\right)&=(\GG^{\hh}_{WY})_{\rho}(S_{\rho}^{\mathbf{a}},S_{\rho}^{\mathbf{b}})=\\
&=\Tr_{\hh}\left(\rho\,\{\mathbf{a},\mathbf{b}\}\right) + \Tr_{\hh}\left(\sqrt{\rho}\,\mathbf{a}\sqrt{\rho}\mathbf{b}\right) - 2\Tr_{\hh}(\rho\mathbf{a})\,\Tr_{\hh}(\rho\mathbf{b}),
\end{split}
\ee
so that
\be
\GG^{\hh}_{WY}\left(\mathbb{W}_{\mathbf{a}},\mathbb{W}_{\mathbf{b}}\right)=\mathbb{W}_{\mathbf{b}}(l_{\mathbf{a}}),
\ee
where $l_{\mathbf{a}}$ is the linear function $l_{\mathbf{a}}(\rho)=\Tr_{\hh}(\rho\mathbf{a})$.
%%%
Then, just as it happens for the Bures-Helstrom metric tensor, since the   vector fields of the type $\mathbb{W}_{\mathbf{a}}$ generate the tangent space at each $\rho\in\stsp(\hh)$, and since the differential of linear functions on $\stsp(\hh)$ generate the cotangent space at each point $\rho\in\stsp(\hh)$, we conclude\footnote{See also \cite{Ciaglia-2020, C-J-S-2020} for alternative proofs of this instance.} that the gradient vector field  (in the sense of Riemannian geometry) associated with the linear function  $l_{\mathbf{a}}$ by means of the Wigner-Yanase metric tensor $\GG^{\hh}_{WY}$  is precisely the vector field  $\mathbb{W}_{\mathbf{a}}$.
%%%

Another similarity between  $\GG_{WY}^{\hh}$ and $\GG_{BH}^{\hh}$ concerns their relation with the geometry of a suitably big sphere endowed with the round metric.
%%%%
However, in the case of the Wigner-Yanase metric tensor, this relation can be seen as the quantum counterpart of a similar property possessed by the Fisher-Rao metric tensor.
%%%%
In particular, it is easy to see that the Fisher-Rao metric tensor on the open interior of the unit simplex $\Delta_{n}$ can be seen as the pull-back of the standard round metric on the n-sphere with radius $r=\frac{1}{4}$ with respect to the square-root map $\mathbf{p}\mapsto\sqrt{\mathbf{p}}$ given by $p^{j}\mapsto\sqrt{p^{j}}$.
%%%%%
Then, it turns out that $\GG_{WY}$ is the pull-back of the standard round metric on the unit sphere in $\bh$ with respect to the square-root map $\rho\mapsto\sqrt{\rho}$ \cite{G-I-2001}.
%%%%
Since the geometrical properties of the n-sphere with the round metric are well-known, with the square-root map at our disposal, we can compute the geodesics, the Riemann curvature, the Ricci tensor, the scalar curvature, and the sectional curvature for $(\stsp(\hh),\GG_{WY}^{\hh})$, and thus understand  all the geometrical features of this Riemannian manifold.
%%%%

\subsection{The Bogoliubov-Kubo-Mori metric tensor}\label{subsec: B-K-M metric tensor}

The last example of monotone Riemannian metric tensor we consider is the so-called    Bogoliubov-Kubo-Mori metric tensor $\GG_{BKM}^{\hh}$ \cite{F-M-A-2019,N-V-W-1975,Naudts-2021,Petz-1994,P-T-1993}.
%%%%%
It corresponds to the function 
\be
f(x)=\frac{x-1}{\ln(x)}=\int_{0}^{1}\,x^{t}\mathrm{d}t
\ee
so that  equation \eqref{eqn: Petz metric 2} becomes
\be
f(L_{\rho}\,R_{\rho^{-1}})\,R_{\rho}=\int_{0}^{1}\,(L_{\rho})^{t}\,(R_{\rho})^{1-t}\mathrm{d}t  \equiv C_{\rho},
\ee
and thus, writing $\GG_{f}^{\hh}=\GG_{BKM}^{\hh}$, equation \eqref{eqn: Petz metric 1} leads us to
\be\label{eqn: BKM metric tensor}
(\GG^{\hh}_{BKM})_{\rho}(\mathbf{v}_{\rho},\mathbf{w}_{\rho})\,=\,\Tr_{\hh}\left(\mathbf{v}_{\rho}\, C^{-1}_{ \rho }(\mathbf{w}_{\rho}) \right)\,,
\ee
where $C_{\rho}(\mathbf{a})=[\rho,\mathbf{a}]=(\rho\mathbf{a} - \mathbf{a}\rho)$, and where $\mathbf{v}_{\rho},\mathbf{w}_{\rho}\in\mathcal{B}_{sa}^{0}(\hh)\cong T_{\rho}\stsp(\hh)$.
%%%%%

Clearly, equation \eqref{eqn: BKM metric tensor} implies that, if we want to compute the Bogoliubov-Kubo-Mori inner product in the linear identification, we need to invert the operator $C_{\rho}$ which may be computationally difficult.
%%%%
However, once again, it is possible to bypass this step if we select the appropriate identification for tangent vectors.
%%%
In this case, it turns out that the  correct one  is the exponential identification  introduced in equation \eqref{eqn:  exponential identification 3} in terms of which it is immediate to check that  equation \eqref{eqn: BKM metric tensor} becomes
\be\label{eqn: BKM metric tensor 2}
(\GG^{\hh}_{BKM})_{\rho}(E^{\mathbf{a}}_{\rho},E^{\mathbf{b}}_{\rho})=\int_{0}^{1}\mathrm{d}\lambda\,\Tr_{\hh}\left( \rho^{\lambda}\,\mathbf{a} \,\rho^{1-\lambda} \mathbf{b} \right) - \Tr_{\hh}\left(\rho\mathbf{b} \right)\Tr(\rho\,\mathbf{a}).
\ee
%%%%
%%%%
Then, it follows from equation   \eqref{eqn: BKM metric tensor 2} that 
\be
\begin{split}
\left(\GG_{BKM}\right)_{\rho}\left(\mathbb{Z}_{\mathbf{a}}(\rho),\mathbb{Z}_{\mathbf{b}}(\rho)\right)&=(\GG^{\hh}_{BKM})_{\rho}(E^{\mathbf{a}}_{\rho},E^{\mathbf{b}}_{\rho})= \\
&= \int_{0}^{1}\,\mathrm{d}\lambda\,\Tr_{\hh}\left(\rho^{\lambda}\,\mathbf{a} \,\rho^{1-\lambda}\mathbf{b}\right) - \Tr_{\hh}(\rho\,\mathbf{a})\,\Tr_{\hh}(\rho\mathbf{b}),
\end{split}
\ee
which means
\be
\GG_{BKM} \left(\mathbb{Z}_{\mathbf{a}} ,\mathbb{Z}_{\mathbf{b}}\right)=\mathbb{Z}_{\mathbf{b}}(l_{\mathbf{a}})
\ee
where $l_{\mathbf{a}}$ is the linear function $l_{\mathbf{a}}(\rho)=\Tr_{\hh}(\rho\mathbf{a})$.
%%%
Therefore, as before, we conclude that $\mathbb{Z}_{\mathbf{a}}$ is the gradient vector field (in the sense of Riemannian geometry) associated with $l_{\mathbf{a}}$ by means of the  Bogoliubov-Kubo-Mori  metric tensor.
%%%%%
Moreover, it is worth noting that the integral curves of $\mathbb{Z}_{\mathbf{a}}$ are geodesics of the so-called exponential connection (compare equation \eqref{eqn:  exponential identification 2} with equation  34 in  \cite{F-M-A-2019}).
%%%%
This instance is particularly intriguing because it “promotes” the integral curves of gradient vector fields to geodesics of a connection which has a role in information geometry, thus giving a first hint to understand the role of  gradient vector fields in information geometry.
%%%%
Moreover, we are naturally lead to ask if something similar happens also for the gradient vector fields of the Bures-Helstrom metric tensor and of the Wigner-Yanase metric tensor, but we will leave this investigation to future work.
%%%%%%

\section{From relative entropies to monotone metrics}\label{sec: from rel entropies to mono metrics}

It is a remarkable fact in classical information geometry that the Fisher-Rao metric tensor can be obtained by a suitable second-order expansion of the Kullback-Leibler relative entropy \cite{K-L-1951}.
%%%%
In a certain sense, the Kullback-Leibler relative entropy is an asymmetric measure of how much two given probability distributions differ from each other with respect to certain information-theoretic tasks like hypothesis testing \cite{K-L-1951}.
%%%%
To roughly explain how the Fisher-Rao metric tensor can be extracted from the Kullback-Leibler relative entropy, let us consider the classical case of probability distributions on a discrete, n-element set $\mathcal{X}_{n}$.
%%%
The set of all such probability distributions is the n-simplex $\overline{\Delta_{n}}$ introduced in equation \ref{eqn: the simplex}.
%%%%
The Kullback-Leibler relative entropy $D_{KL}$ between $\mathbf{p}$ and $\mathbf{q}$ reads
\be
D_{KL} (\mathbf{p},\mathbf{q})=\sum_{j=1}^{n}\,p^{j}\ln(p^{j}) - p^{j}\ln(q^{j})\,.
\ee
%%%
This quantity is well-defined  when $q^{j}=0$ implies $p^{j}=0$  and we exploit
\be
\lim_{x\ra 0^{+}}\,x\ln(x)=0.
\ee
%%%%
In particular, if we focus on probability distributions lying in the interior $\Delta_{n}$ of the simplex, $D_{KL}$ is well-defined and smooth on $\Delta_{n}\times\Delta_{n}$ (recall that $\Delta_{n}$ is a smooth manifold because it is the intersection of the open positive hyperoctant with the hyperplane defined by $\sum_{j=1}^{n}p^{j}=1$).
%%%%
Now,  if we define $g_{jk}$ to be
\be\label{eqn: KL to FR}
g_{jk}\,:=\,-\left(\frac{\partial^{2}}{\partial p^{j}\partial q^{k}}\,D_{KL}\right)_{\mathbf{p}=\mathbf{q}}\,=\,\delta_{jk}\frac{1}{p^{j}},
\ee
it immediately follows that we can express the Fisher-Rao metric tensor $\GG_{FR}^{n}$ on $\Delta_{n}$ as
\be\label{eqn: KL to FR 2}
\GG_{FR}^{n} \,=\, g_{jk}\,\mathrm{d}p^{j}\otimes\mathrm{d}p^{k}.
\ee
%%%
Equations \eqref{eqn: KL to FR} and \eqref{eqn: KL to FR 2}  explain in which sense $\GG_{FR}^n $ may be considered as  a sort of second-order expansion of $D_{KL} $.
%%%%
This seemingly ad-hoc procedure may be actually reformulated to produce a coordinate-free algorithm that works for a particular class of smooth functions on $M\times M$ where $M$ is a smooth manifold, and, in section \ref{sec: from 2-functions to covariant tensors}, we briefly recall this approach as developed in  \cite{C-DC-L-M-M-V-V-2018, M-M-V-V-2017} in order to keep the present exposition as much self-contained as possible.
%%%
Moreover, it is possible to lift this coordinate-free procedure to an even more sophisticated layer by using Lie groupoids and their associated Lie algebroids \cite{G-G-K-M-2019,G-G-K-M-2020}.
%%%%%
 
%%%%

The Kullback-Leibler relative entropy is not the only possible notion of relative entropy available in classical information theory and statistics.
%%%%
For instance, there is a family of relative entropies known as f-divergences \cite{A-S-1966,Csizar-1963,Morimoto-1963} which is indexed by a convex function  $f\colon\mathbb{R}^{+}_{0}\ra\mathbb{R}^{+}_{0}$ (satisfying $f(1)=0$) and whose members, in the discrete case, are given by
\be
D_{f} (\mathbf{p},\mathbf{q}):=\sum_{j=1}^{n}\,f\left(\frac{p^{j}}{q^{j}}\right)\,q^{j}.
\ee
%%%
Note that $D_{KL} =D_{f} $ with $f(x)=x\ln(x)$.
%%%%
This family is particularly relevant because it is the only family whose members are such that
\be\label{eqn: monotonicity of f-divergences}
D_{f}(\mathbf{p},\mathbf{q})\geq D_{f} (\Phi(\mathbf{p}),\Phi(\mathbf{q}))
\ee
for every  classical Markov map  $\Phi$ from $\Delta_{n}$ to $\Delta_{m}$.
%%%%
Quite interestingly, if we define $\GG_{jk}^{f}$ to be
\be\label{eqn: f to FR}
\GG_{jk}^{f}\,:=\,-\left(\frac{\partial^{2}}{\partial p^{j}\partial q^{k}}\,D_{f}\right)_{\mathbf{p}=\mathbf{q}}\,=\,\delta_{jk}\frac{f''(1)}{p^{j}},
\ee
it immediately follows that
\be\label{eqn: f to FR 2}
f''(1)\,\GG_{FR}^n \,=\,\GG_{jk}^{f}\,\mathrm{d}p^{j}\otimes\mathrm{d}p^{k},
\ee
which means that the second-order expansion (in the sense explained above) of an f-divergence is a constant multiple of the Fisher-Rao metric tensor.
%%%%
This  result is not completely surprising if we recall equation \eqref{eqn: monotonicity of f-divergences} and the fact that the family of Fisher-Rao metric tensors is the only family of classical monotone metric tensors  (up to an overall multiplicative constant).
%%%%

Given this classical picture, it is natural to ask what is its quantum counterpart.
%%%%
It turns out that also in this case there are different notions of quantum relative entropies (quantum divergences) satisfying a monotonicity property with respect to quantum Markov maps.
%%%
However, recalling the non-uniqueness of monotone metric tensors arising in Petz's classification, it turns out that different quantum divergences lead, in principle, to different monotone metric tensors.
%%%%
For instance, if we consider the von Neumann-Umegaki relative entropy \cite{Umegaki-1962-4}
\be
S_{vNU}(\rho,\sigma)=\Tr_{\hh}\left(\rho\ln(\rho)- \rho\ln(\sigma)\right),
\ee
which is evidently a formal analogue of the Kullback-Leibler relative entropy, it is possible to prove \cite{Ciaglia-2020, F-M-A-2019, Petz-1994} that the monotone metric tensor we can extract from it is the Bogoliubov-Kubo-Mori metric tensor.
%%%%%
Moreover, if we consider the so-called quantum fidelity for mixed states \cite{Cantoni-1975, F-J-R-2016, Josza-1994, Uhlmann-1976, Uhlmann-2011} (also known as Bures metric \cite{Bures-1969}) given by
\be
S_{B}(\rho,\sigma)\,=\left[\Tr_{\hh}\left(\sqrt{\sqrt{\rho}\,\sigma\sqrt{\rho}}\right)\right]^{2},
\ee
it is possible to prove \cite{C-DC-L-M-M-V-V-2018, Uhlmann-1992} that the monotone metric tensor we extract from $S_{B}$ is the Bures-Helstrom metric tensor reviewed in subsection \ref{subsec: B-H metric tensor}.
%%%%%

An  important family of quantum relative entropies is the family of $\alpha$\textit{-R\'enyi relative entropies} ($\alpha$-RREs) \cite{D-L-2014, Tomamichel-2016}
\be
\widetilde{D}_{\alpha}(\rho,\varrho)=\frac{1}{\alpha-1}\log\text{Tr}\,\bigl(\rho^q\varrho^{1-q}\bigr),
\ee
where $\alpha\in(0,1)\cup(1,\infty)$, whose members allow to describe the cut-off rates in quantum binary state discrimination \cite{M-H-2011}. 
%%%%%

A sort of non-commutative version of the $\alpha$-RREs is given by the family of  $\alpha$-\textit{quantum R\'enyi divergences} ($\alpha$-QRDs) \cite{ML-D-S-F-T-2013} - also known as  \textit{sandwiched Rényi $\alpha$-divergences} - given by
\be \label{Sqrd}
D_{\alpha}(\rho,\varrho)=\frac{1}{\alpha-1}\log\text{Tr}\,\bigl(\varrho^{\frac{1-\alpha}{2\alpha}}\rho\varrho^{\frac{1-\alpha}{2\alpha}}\bigr)^{\alpha},
\ee
where again $\alpha\in(0,1)\cup(1,\infty)$.
%%%%
Indeed, it is clear that $D_{\alpha}(\rho,\varrho)=\widetilde{D}_{\alpha}(\rho,\varrho)$ whenever $[\rho,\,\varrho]=0$.
%%%%%
However,  it turns out that the $\alpha$-QRDs do not satisfy the monotonicity property
\be 
D_{\alpha}(\Phi(\rho)|\Phi(\varrho))\leq D_{q}(\rho|\varrho) 
\ee
with respect to every CPTP map $\Phi$   for $\alpha\in(0,1/2)$ \cite{T-F-2017}.
%%%%
The family of metric tensors associated with the family of $\alpha$-QRDs  was thoroughly investigated in \cite{M-M-V-V-2017}.
%%%%%%

Both the  $\alpha$-RREs and the $\alpha$-QRDs can be looked at as one-parameter subfamilies of a two-parameter family of quantum relative entropies known as $\alpha$-$z$-\textit{R\'enyi Relative Entropies} ($\alpha$-$z$-RREs)  \cite{A-D-2015,J-O-P-P-2012} and given by
\be 
D_{\alpha,z}(\rho,\varrho)=\frac{1}{q-1}\log\text{Tr}\,\bigl(\rho^{\frac{\alpha}{2z}}\varrho^{\frac{1-\alpha}{z}}\rho^{\frac{\alpha}{2z}}\bigr)^z\,.
\ee
%%%%
According to \cite{A-D-2015}, it holds
\be 
\begin{split}
\widetilde{D}_{\alpha}(\rho,\varrho)&=\lim_{z\to1} D_{\alpha,z}(\rho,\varrho)\equiv =\frac{1}{q-1}\log\text{Tr}\,\bigl(\rho^q\varrho^{1-q}\bigr) \\
D_{\alpha}(\rho,\varrho)&=\lim_{z\to \alpha} D_{\alpha,z}(\rho,\varrho),
\end{split}
\ee
and we also obtain the  von Neumann-Umegaki relative entropy as
\be 
S_{vNU}(\rho,\sigma)=\lim_{z=\alpha\to 1} D_{\alpha,z}(\rho,\varrho).
\ee
%%%%%
The monotonicity property under CPTP maps of the  the  $\alpha$-$z$-RREs  was studied in \cite{C-F-L-2016} and completely characterized in \cite{Zhang-2020}, while the family of metric tensors generated by this family of quantum relative entropies was studied in \cite{C-DC-L-M-M-V-V-2018}. 
%%%%

According to the results recalled in section \ref{sec: from 2-functions to covariant tensors},  whenever a quantum relative entropy satisfies the monotonicity property under CPTP maps, the associated Riemannian metric tensor falls into the class of monotone metric tensors classified by Petz.
%%%%
Consequently, both the $\alpha$-QRDs and the $\alpha$-$z$-RREs lead to monotone metric tensors, but it is not clear if all of these metric tensors can be recovered from these families of quantum relative entropies.
%%%%%

On the other hand, there is a family of  quantum relative entropies  defined by Petz on arbitrary von Neumann algebras in \cite{Petz-1985}, and then thoroughly analysed in  the finite-dimensional setting in \cite{Petz-1986}, that essentially encodes all the monotone metric tensors on $\stsp(\hh)$.
%%%%
To define this family of quantum relative entropies in the case we are interested in, we first need to select a  continuous function $g$ on $(0,\infty)$.
%%%%
Then, to any such function, we associate a \grit{relative g-entropy} $S_{g}^{\hh}$ according to 
\be\label{def: entropy Lev Ruskai}
S_{g}^{\hh}(\rho,\sigma)=\langle\sqrt{\rho},g(L_{\sigma}R_{\rho^{-1}})(\sqrt{\rho})\rangle^{HS}_{\hh}=\Tr_{\hh}\left(\sqrt{\rho}\,g(L_{\sigma}R_{\rho^{-1}})(\sqrt{\rho})\right).
\ee
%%%%
When $-g$ is operator monotone, so that $g$ is operator convex, it turns out that 
\be
S_{g}^{\hh}(\rho,\sigma)\geq S_{g}^{\mathcal{K}}(\Phi(\rho),\Phi(\sigma))
\ee
for all CPTP maps $\Phi\colon\bh\ra\mathcal{B}(\mathcal{K})$ \cite[thm. 4]{Petz-1986}.
%%%%

The metric tensors on $\stsp(\hh)$ associated with this family of quantum relative entropies have been studied in \cite{L-R-1999}, where it was shown that, under suitable additional conditions on $g$, they coincide with the monotone metric tensors classified by Petz \cite{Petz-1996} and briefly discussed in section \ref{sec: Cencov and Petz}.
%%%%%%
In subsection \ref{subsec: unfolding rel entropies and mono metrics} we will give an alternative proof of the results  in \cite{L-R-1999} which is based on the unfolding procedure presented in \cite{C-DC-L-M-M-V-V-2018, M-M-V-V-2017} and thoroughly discussed in the whole section \ref{sec: unfolding of quantum states}.
%%%%

\section{From two-point functions to covariant tensor fields}\label{sec: from 2-functions to covariant tensors}

In this section, we will describe the procedure used in \ref{subsec: unfolding rel entropies and mono metrics} to extract a symmetric, covariant tensor field of type $(0,2)$ on a smooth manifold $\mathcal{M}$ from a function $f \in \mathcal{F} (\mathcal{M}\times\mathcal{M})$, which we refer to as a \emph{two-point function}. 
%%%%
The aim is to give an intrinsic version of the expression given in equation \eqref{eqn: KL to FR}.
%%%%
Roughly speaking, we will achieve this result by considering Lie derivatives of a two point function with respect to suitable lifts on $\mathcal{M}\times\mathcal{M}$ of vector fields  on $\mathcal{M}$. 
%%%%
Some of the results will be presented without proof for which  we refer the reader to \cite{C-DC-L-M-M-V-V-2018}.
%%%%%

Let us stress that the  procedure we are about to present is by no means unique, and different approaches to the same construction have been proposed. 
%%%%
For instance, in \cite{M-M-V-V-2017} an analogous construction is performed in terms of bi-forms on $\mathcal{M} \times \mathcal{M}$, while in \cite{G-G-K-M-2019,G-G-K-M-2020} the fact that $\mathcal{M} \times \mathcal{M}$ is a \emph{pair groupoid} of which $T\mathcal{M}$ is the associated \emph{Lie algebroid} is exploited.
%%%%%

\vsp

Given a smooth manifold $\mathcal{M}$, there exist two canonical projections from $\mathcal{M}\times\mathcal{M}$ to $\mathcal{M}$ given by
\begin{equation}\label{def: r and l projections M}
\begin{split}
\pi_l &: \mathcal{M}\times\mathcal{M} \ni (m_1,m_2) \mapsto m_1 \in \mathcal{M}, \\
\pi_r &: \mathcal{M}\times\mathcal{M} \ni (m_1,m_2) \mapsto m_2 \in \mathcal{M}.
\end{split}
\end{equation}
%%%%%
These projection maps allow us to define two particular classes of vector fields on $\mathcal{M}\times\mathcal{M}$ that can be thought of as ``lifted'' vector fields in a sense that will be clear in a moment.

\begin{definition} \label{def: l and r lifts}
Let  $\mathfrak{X}(\mathcal{M})$ denote the space of smooth vector fields on $\mathcal{M}$, and let $X \in \mathfrak{X}(\mathcal{M})$.
%%%%
We define $\mathbb{X}_l \in \mathfrak{X}(\mathcal{M} \times \mathcal{M})$ as the vector field which is $\pi_l$-related to $X$ and $\pi_r$-related to the null vector field $\mathbf{0} \in \mathfrak{X}(\mathcal{M})$, that is, $ \mathbb{X}_l $ is such that
	\begin{equation}\label{eqn: left lift}
	\begin{split}
	T \pi_l \circ \mathbb{X}_l & = X \circ \pi_l, \\
	T \pi_r \circ \mathbb{X}_l & = \mathbf{0} \circ \pi_r.
	\end{split}
	\end{equation}
We call $\mathbb{X}_{l}$ the \grit{left lift} of $X$.
%%%%
Analogously, the \grit{right lift}	$\mathbb{X}_r$ of $X$ is defined as the vector field which is $\pi_r$-related to $X$ and $\pi_l$-related to the null vector field $\mathbf{0} \in \mathfrak{X}(\mathcal{M})$, that is, $ \mathbb{X}_r $ is such that
	\begin{equation}\label{eqn: right lift}
	\begin{split}
	T \pi_r \circ \mathbb{X}_r & = X \circ \pi_r, \\
	T \pi_l \circ \mathbb{X}_r & = \mathbf{0} \circ \pi_l.
	\end{split}
	\end{equation}
In proposition \ref{eqn: generic expression of vector field wrt product coordinate chart}, it will be proved that both the \grit{left} and \grit{right lift} of $X$ are uniquely defined.
%%%%
\end{definition}
In order to get a better understanding of the role of the lifted vector fields let us introduce a coordinate chart $(x^j, y^k)$ with $j,k=1,2,\dots,dim (\mathcal{M})$ on $\mathcal{M}\times\mathcal{M}$ adapted to its product structure.
%%%
Specifically, this means that both $(x^j)$ and $(y^k)$ are coordinate charts on $\mathcal{M}$, so that specifying both coordinates spells out a point in $\mathcal{M} \times \mathcal{M}$. 
%%%%
In this chart, a vector field $Y\in \mathfrak{X}(\mathcal{M}\times\mathcal{M})$ can be written as
\begin{equation}\label{eqn: generic expression of vector field wrt product coordinate chart}
Y = Y_x^j \frac{\partial}{\partial x^j} +  Y_y^k \frac{\partial}{\partial y^k}
\end{equation}
and we have the following proposition.
\begin{proposition} \label{prop: properties of lifted fields}
In the coordinate chart $(x^j, y^k)$ introduced above,  the vector fields $\mathbb{X}_l$ and $\mathbb{X}_r$ on $\mathcal{M} \times \mathcal{M}$ associated with the vector field $X$ on $\mathcal{M}$ according to definition \ref{def: l and r lifts} have the expressions
	\begin{equation} \label{eqn: coord expr X_l and X_r}	
	\begin{split}
		\mathbb{X}_l (x,y) & = X^j(x) \frac{\partial}{\partial x^j}, \\
		\mathbb{X}_r (x,y) & = X^k(y) \frac{\partial}{\partial y^k}.
		\end{split}
	\end{equation}
Moreover, given any $X,Y \in \mathfrak{X}(\mathcal{M})$ and $f\in \mathcal{F}(\mathcal{M})$,  setting $f_l = \pi_l^* f$ and $f_r = \pi_r^* f$, and denoting with $L$ the Lie derivative, the following equalities hold:
	\begin{align} \label{eqn: properties lift 1}
	[\mathbb{X}_l,\mathbb{Y}_l] &= [X,Y]_l, & [\mathbb{X}_r,\mathbb{Y}_r] &= [X,Y]_r, & [\mathbb{X}_r,\mathbb{Y}_l] &= [\mathbb{X}_l,\mathbb{Y}_r] = 0, \\
	(fX)_l &= f_l \mathbb{X}_l, & (fX)_r &= f_r \mathbb{X}_r, & L_{\mathbb{X}_l}f_r &= L_{\mathbb{X}_r}f_l = 0. \label{eqn: properties lift 2}
	\end{align}
\end{proposition}
\begin{proof}
According to equation \eqref{eqn: generic expression of vector field wrt product coordinate chart}, the local expressions of the vector fields $\mathbb{X}_l$ and $\mathbb{X}_r$ in the coordinates $(x^j, y^k)$ read 
	\begin{equation}
		\begin{split}
			\mathbb{X}_l & = (\mathbb{X}_{l})^{j}_{x} \frac{\partial}{\partial x^j} +  (\mathbb{X}_{l})^{k}_{y}  \frac{\partial}{\partial y^k}, \\
			\mathbb{X}_r & = (\mathbb{X}_{r})^{j}_{x} \frac{\partial}{\partial x^j} +  (\mathbb{X}_{r})^{k}_{x} \frac{\partial}{\partial y^k},
		\end{split}
	\end{equation}
where the functions $ (\mathbb{X}_{l})^{j}_{x}, (\mathbb{X}_{l})^{k}_{y} $  are determined by the two conditions in equation \eqref{eqn: left lift}, while the functions $(\mathbb{X}_{r})^{j}_{x},(\mathbb{X}_{r})^{k}_{x}$ are determined by the two conditions in equation \eqref{eqn: right lift}.
%%%%
In particular,  since $\pi_{l}^{j}(m)=x^{j}(m)$, a direct computation shows that the first condition in equation \eqref{eqn: left lift} implies
	\begin{equation}
		(\mathbb{X}_{l})^{j}_{x}  = X^j ,
	\end{equation}
while the   second condition in  equation \eqref{eqn: left lift} implies
	\begin{equation}
		(\mathbb{X}_{l})^{k}_{y} =0.
	\end{equation}
%%%%
A completely analogous computation can be carried out for $\mathbb{X}_r$, and we proved the validity of equation \eqref{eqn: coord expr X_l and X_r}.
%%%%
Moreover, let us notice that this instance also shows that definition \ref{def: l and r lifts} uniquely determines the vector fields $\mathbb{X}_l$ and $\mathbb{X}_r$ given the vector field $X$ on $\mathcal{M}$.
%%%%	
The proof of the results in equation \eqref{eqn: properties lift 1} and \eqref{eqn: properties lift 2} can be now obtained with straightforward computations using the local expressions of $\mathbb{X}_{l}$ and $\mathbb{X}_{r}$ given in equation \eqref{eqn: coord expr X_l and X_r} and thus is omitted.
%%%%

\end{proof}

Besides the canonical projections $\pi_{l}$ and $\pi_{r}$, the procedure we want to describe makes use of the so-called  \emph{diagonal immersion} given by the map
\begin{equation} \label{def: diag immersion}
i_D:\mathcal{M}\ni m \mapsto (m,m) \in \mathcal{M}\times\mathcal{M}.
\end{equation}
It is immediate to check that the local expression with respect to the coordinate chart $(x^j, y^j)$ introduced above of the operation of taking the pull-back with respect to $i_{D}$ coincides with setting $x^{j}=y^{j}$.
%%%% 
At this point, we are ready to give  the coordinate-free counterpart of the procedure appearing in equation \eqref{eqn: KL to FR}.
\begin{proposition}\label{prop: properties of g}
Let $S\in\mathcal{F}(\mathcal{M}\times\mathcal{M})$ be smooth, and  take $X,Y\in\mathfrak{X}(\mathcal{M})$.
%%%%
Then, the expressions
	\begin{align}
	g_{ll}(X,Y) & := i^*_D\left(L_{\mathbb{X}_l}L_{\mathbb{Y}_l} S \right), & g_{rr}(X,Y) & := i^*_D\left(L_{\mathbb{X}_r}L_{\mathbb{Y}_r} S \right), \\
	g_{lr}(X,Y) & := i^*_D\left(L_{\mathbb{X}_l}L_{\mathbb{Y}_r} S \right), & g_{rl}(X,Y) & := i^*_D\left(L_{\mathbb{X}_r}L_{\mathbb{Y}_l} S \right) 
	\end{align}
are such that
	\begin{enumerate}
		\item $g_{lr}$ and $g_{rl}$ are covariant $(0,2)$-type tensors on $\mathcal{M}$ such that 
		\begin{equation}
		g_{lr} (X,Y) = g_{rl} (Y,X);
		\end{equation}
		\item $g_{ll}$ is a symmetric covariant $(0,2)$-type tensors on $\mathcal{M}$ if and only if
		\begin{equation} \label{eqn: Lie der w.r.t. Xl of S}
		i^*_D(L_{\mathbb{X}_l}S) = 0 \quad \forall X \in \mathfrak{X}(\mathcal{M});
		\end{equation} 
		\item $g_{rr}$ is a symmetric covariant $(0,2)$-type tensors on $\mathcal{M}$ if and only if
		\begin{equation} \label{eqn: Lie der w.r.t. Xr of S}
		i^*_D(L_{\mathbb{X}_r}S) = 0 \quad \forall X \in \mathfrak{X}(\mathcal{M});
		\end{equation} 
		\item if $S$ satisfies both \eqref{eqn: Lie der w.r.t. Xl of S} and \eqref{eqn: Lie der w.r.t. Xr of S}, it holds
		\begin{equation}
		g := g_{ll} = g_{rr} = - g_{lr} = - g_{rl},
		\end{equation}
		and all these tensors are symmetric.
		Moreover,  the local expression of $g$ reads
		\be\label{eqn: local expression of 2-tensor}
		\begin{split}
		g&=\left(\frac{\partial^{2} S}{\partial x^{j}\partial x^{k}}\right)_{\mathbf{x}=\mathbf{y}}\mathrm{d}x^{j}\otimes\mathrm{d}x^{k}=  -\left(\frac{\partial^{2} S}{\partial x^{j}\partial y^{k}}\right)_{\mathbf{x}=\mathbf{y}}\mathrm{d}x^{j}\otimes\mathrm{d}x^{k} =\\
&=\left(\frac{\partial^{2} S}{\partial y^{j}\partial y^{k}}\right)_{\mathbf{x}=\mathbf{y}}\mathrm{d}x^{j}\otimes\mathrm{d}x^{k}
		=-\left(\frac{\partial^{2} S}{\partial y^{j}\partial x^{k}}\right)_{\mathbf{x}=\mathbf{y}}\mathrm{d}x^{j}\otimes\mathrm{d}x^{k},
		\end{split}
		\ee
		where $(x^{j},y^{k})$ is the product coordinate chart introduced above.
	\end{enumerate}
\end{proposition}
\begin{proof}
The proof consists of a (possibly tedious) direct check, and we refer the reader to the proofs of Propositions 3 and 5 in \cite{C-DC-L-M-M-V-V-2018} for all the explicit details.
\end{proof}

Looking back at equation  \eqref{eqn: KL to FR}, we see that the Kullback-Leibler relative entropy is a two-point function on $\Delta_{n}$ satisfying both \eqref{eqn: Lie der w.r.t. Xl of S} and \eqref{eqn: Lie der w.r.t. Xr of S}, and thus the fact that equation  \eqref{eqn: KL to FR} defines the coefficients of a metric tensor on $\Delta_{n}$ is assured by proposition \ref{prop: properties of g} above.
%%%%

Motivated by proposition \ref{prop: properties of g} above, we introduce the notion of \textbf{potential function} and \textbf{divergence function} as follows.
%%%%

\begin{definition} \label{def: potential functions}
A smooth, two-point function $S$ on $\mathcal{M}$ is called  a \textbf{potential function} if it satisfies equation \eqref{eqn: Lie der w.r.t. Xl of S} and equation \eqref{eqn: Lie der w.r.t. Xr of S}.
%%%%

\end{definition}

Clearly, a  \textbf{potential function} on $\mathcal{M}$ immediately leads to a possibly singular, pseudo-Riemannian tensor on it, that is, a symmetric, covariant tensor of type $(0,2)$ whose signature need not be positive.
%%%%
In almost all the cases considered in Classical and Quantum Information Geometry, the two-point functions under consideration are relative entropies or distance functions satisfying a positivity property.
%%%%
This instance leads us to introduce the notion of  \textbf{divergence function} as follows.
%%%%

\begin{definition} \label{def: divergence functions}
A \textbf{divergence function} is a smooth, two-point function $S$ on $ \mathcal{M}$ satisfying
	\begin{equation}
		S(m_1,m_2) \ge 0, \quad \quad S(m,m) = 0  		
	\end{equation}
for all $m,m_{1},m_{2}\in \mathcal{M}$.
%%%

%%%%
\end{definition}

\begin{proposition}\label{prop: divergence functions lead to positive semidefinite tensors}
Every \textbf{divergence function} $S$ on $\mathcal{M}$ is also a \textbf{potential function} on $\mathcal{M}$, and the covariant tensor extracted from $S$ according to proposition \ref{prop: properties of g} is positive semidefinite.
%%%%
\end{proposition}
\begin{proof}
It is clear that the diagonal part of $\mathcal{M}\times\mathcal{M}$ is a submanifold of critical points for $S$, specifically, of local minima.
%%%%
Therefore, the fact that $S$ is a \textbf{potential function} follows upon comparing the local expression of a critical point for $S$  with the coordinate expressions of equation \eqref{eqn: Lie der w.r.t. Xl of S} and equation \eqref{eqn: Lie der w.r.t. Xr of S} in the coordinate chart $(x^j, y^j)$ introduced above.
%%%%

To prove that $g$ is positive semidefinite, let us define the function
\be
F(t):=S(\phi_{t}^{X}(m),m),
\ee
where $\phi_{t}^{X}(m)$ denotes the flow of the vector field $X$ starting at $m\in M$.
%%%%
Clearly, $F\geq 0$ for all $t\in(-\epsilon,\epsilon)$ with $\epsilon >0$, and $F(0)=0$.
%%%%
Referring to the product coordinate chart $(x^{j},y^{k})$ introduced above, and taking the second derivative of $F$ at $t=0$, we obtain
%\be
%\left.\frac{\mathrm{d}^{2}F}{\mathrm{d}t^{2}}\right|_{t=0}=\left(\frac{\mathrm{d} }{\mathrm{d}t }\left(\frac{\partial S}{\partial x^{j}}\frac{\mathrm{d}(\Phi^{X}_{t}(\rho))^{j}}{\mathrm{d}t}\right)\right)_{t=0}
%\ee
\be
\begin{split}
\left.\frac{\mathrm{d}^{2}F}{\mathrm{d}t^{2}}\right|_{t=0}&= \left(\frac{\partial^{2} S}{\partial x^{j}\partial x^{k}}\right)_{\mathbf{x}=\mathbf{y}}X^{j} X^{k}  + \left(\frac{\partial S}{\partial x^{j}}\right)_{\mathbf{x}=\mathbf{y}}\left(\frac{\partial X^{j}}{\partial x^{k}} \right) X^{k}  \\
& = g_{jk} X^{j} X^{k},
\end{split}
\ee
where we used equation \eqref{eqn: Lie der w.r.t. Xl of S} and equation \eqref{eqn: local expression of 2-tensor}.
%%%%
Now, recalling that $t=0$ is a critical point for $F$, if it were $F''(0)<0$ we would conclude that $t=0$ is a local maximum for $F$  which is impossible since $F\geq 0$ and $F(0)=0$.
%%%
Consequently, we must conclude that 
\be
F''(0)=g_{jk} X^{j} X^{k}\geq 0 .
\ee
%%%%
Since this inequality holds for all $m\in M $ and for all vector fields $X$ on $M$, we conclude that $g$ is indeed positive semidefinite.
%%%%
\end{proof}

\subsection{Potential functions and smooth maps }

In this subsection  we shall investigate the behaviour of potential functions and their associated tensors  with respect to smooth mappings. 
%%%%
Understanding this behaviour will be important when we will compare the results of  the unfolding procedure for monotone metrics with the results of the unfolding procedure for quantum divergences in section \ref{sec: unfolding of quantum states}.
%%%%%

%%%%%
Let $\mathcal{M}$ and $\mathcal{N}$ be  smooth manifolds, and let   $\phi$ be a smooth map from $\mathcal{M}$ to $\mathcal{N}$.
%%%%
We define the map $\Phi$ from $\mathcal{M} \times \mathcal{M}$ to $\mathcal{N} \times \mathcal{N}$ as
\begin{equation} \label{def: Phi}
	\Phi: \mathcal{M} \times \mathcal{M} \ni (m_1,m_2) \mapsto (\phi(m_1), \phi(m_2)) \in \mathcal{N} \times \mathcal{N},
\end{equation}
and a direct computation shows that
\begin{equation} \label{eqn: rel imm with phi}
	\Phi \circ i^{\mathcal{M}}_D = i^{\mathcal{N}}_D \circ \phi,
\end{equation}
where $i^{\mathcal{N}}_D$  and $i^{\mathcal{M}}_D$ are, respectively,  the diagonal immersion of $\mathcal{N}$ in $\mathcal{N} \times \mathcal{N}$ and the diagonal immersion of $\mathcal{M}$ in $\mathcal{M} \times \mathcal{M}$.
%%%%

\begin{proposition}\label{prop: X and Z phi-related implies Xl and Zl Phi-related}
Let $X\in\mathfrak{X}(\mathcal{M})$ be $\phi$-related to $Z\in\mathfrak{X}(\mathcal{N})$, that is, $T\phi\circ X=Z\circ \phi$.
%%%%
Then, $\mathbb{X}_{l}$ is $\Phi$-related to $\mathbb{Z}_{l}$, that is, $T\Phi\circ\mathbb{X}_{l}=\mathbb{Z}_{l}\circ \Phi$, and $\mathbb{X}_{r}$ is $\Phi$-related to $\mathbb{Z}_{r}$, that is, $T\Phi\circ\mathbb{X}_{r}=\mathbb{Z}_{r}\circ \Phi$.
%%%%
\end{proposition}
\begin{proof}
Let us start noting that a point in the tangent space   $TM$ can be denoted by $(m\,,v_{m})$, with $m\in \mathcal{M}$ and $v_{m}\in T_{m}\mathcal{M}$, but, since in general $T\mathcal{M}$ is not a Cartesian product,  the notation $(m\,,v_{m})$ should be handled keeping in mind that the second factor  is not independent from the first one.
%%%%
A vector field $X\in\mathfrak{X}(\mathcal{M})$ can be thought of as a derivation of the associative algebra $\mathcal{F}(\mathcal{M})$ of real-valued, smooth functions on $\mathcal{M}$, or as a smooth section of the tangent bundle $T\mathcal{M}$, that is, a smooth map $X\colon \mathcal{M}\rightarrow T\mathcal{M}$ such that $\tau\circ X=id_{\mathcal{M}}$, where $\tau$ is the canonical projection of the tangent bundle.
In the latter case, we may write the evaluation of a  vector field on $m\in \mathcal{M}$ as $X(m)=(m\,,v^{X}_{m})$.
%%%%
Then, a direct check shows that
\be
\begin{split}
T\phi\circ X(m)&=T\phi(m\,,v^{X}_{m})=\left(\phi(m)\,,T_{m}\phi(v^{x}_{m})\right)\,, \\
Z\circ\phi(m)&=Z(\phi(m))=\left(\phi(m)\,,v_{\phi(m)}^{Z}\right)\,,
\end{split}
\ee
from which it follows that $T\phi\circ X=Z\circ \phi$ implies 
\be\label{eqn: potential functions and differentiable mapping, relatedness of vector fields 1}
T_{m}\phi(v^{X}_{m})=v_{\phi(m)}^{Z}\,.
\ee
%%%%
Now, we have
\be
T\Phi(m_{1}\,,v_{m_{1}}\,;m_{2}\,,v_{m_{2}})=\left(\phi(m_{1})\,,T_{m_{1}}\phi(v_{m_{1}})\,;\phi(m_{2})\,,T_{m_{2}}\phi(v_{m_{2}})\right)\,,
\ee
and thus, recalling that $\mathbb{X}_{l}$ is the left lift of $X$, it holds
\be\label{eqn: potential functions and differentiable mapping, relatedness of vector fields 2}
T\Phi\circ\mathbb{X}_{l}(m_{1}\,,m_{2})=T\Phi(m_{1}\,,v^{X}_{m_{1}}\,,m_{2}\,,0)=\left(\phi(m_{1})\,,T_{m_{1}}\phi(v^{X}_{m_{1}})\,;\phi(m_{2})\,,0\right)\,.
\ee
%%%%
On the other hand, recalling that $\mathbb{Z}_{l}$ is the left lift of $Z$, it also holds
\be\label{eqn: potential functions and differentiable mapping, relatedness of vector fields 3}
\mathbb{Z}_{l}\circ\Phi(m_{1}\,,m_{2})=\mathbb{Z}_{l}(\phi(m_{1})\,,\phi(m_{2}))=\left(\phi(m_{1})\,,v_{\phi(m_{1})}^{Z}\,;\phi(m_{2})\,,0\right)\,.
\ee
%%%%
Using equation \eqref{eqn: potential functions and differentiable mapping, relatedness of vector fields 1} in equation \eqref{eqn: potential functions and differentiable mapping, relatedness of vector fields 3}, and then comparing equation \eqref{eqn: potential functions and differentiable mapping, relatedness of vector fields 3} with equation \eqref{eqn: potential functions and differentiable mapping, relatedness of vector fields 2} we obtain 
\be
T\Phi\circ\mathbb{X}_{l}=\mathbb{Z}_{l}\circ\Phi
\ee
as claimed.
%%%%
Proceeding analogously, we prove that $\mathbb{X}_{r}$ is $\Phi$-related to $\mathbb{Z}_{r}$.
%%%%
\end{proof}

With proposition \ref{prop: X and Z phi-related implies Xl and Zl Phi-related} at hand we can prove the following:
\begin{proposition} \label{prop: S under smooth maps}
Let $\phi$ be a smooth map from $\mathcal{M}$ to $\mathcal{N}$ and $\Phi$ defined as in \eqref{def: Phi}.
%%%%
If $S$ is a potential function on $\mathcal{N} \times \mathcal{N}$ then $\Phi^* S$ is a potential function on $\mathcal{M} \times \mathcal{M}$, moreover, the converse holds if $\phi$ (and thus $\Phi$) is surjective.
%%%%
\end{proposition}
\begin{proof}
Let $X$ be a vector field on $\mathcal{M}$ and $Z$ a vector field on $\mathcal{N}$ which is $\phi$-related to $X$. Being $S$ a potential function (on $\mathcal{N} \times \mathcal{N}$) we have
\be 
(i^{\mathcal{N}}_D)^*(L_{\mathbb{Z}_l} S) = 0,
\ee
then
\be\label{eqn: S under smooth maps 1}
\begin{split}
0 &=  (i^{\mathcal{N}}_D \circ \phi)^*(L_{\mathbb{Z}_l} S)   = (\Phi \circ i^{\mathcal{M}}_D)^*(L_{\mathbb{Z}_l} S) = (i^{ \mathcal{M} }_D )^* \circ \Phi^* (L_{\mathbb{Z}_l} S) =  (i^{ \mathcal{M} }_D )^* ( L_{\mathbb{X}_l} 	\Phi^* S ),
\end{split}
\ee
where we used \eqref{eqn: rel imm with phi}, proposition \ref{prop: X and Z phi-related implies Xl and Zl Phi-related} and the fact that 
\be\label{eqn: lie derivative and pushforward}
\Phi^* (L_{\mathbb{Z}_l} S) =  L_{\mathbb{X}_l} 	\Phi^* S  ,
\ee
because $\mathbb{X}_l$ and $\mathbb{Z}_l$ are $\Phi$-related according to proposition \ref{prop: X and Z phi-related implies Xl and Zl Phi-related} (see proposition 4.2.8 in \cite{A-M-R-1988} for a proof of the validity of equation \eqref{eqn: lie derivative and pushforward}).
%%%
An analogous computation can be carried out to conclude that $(i^{ \mathcal{M} }_D )^* ( L_{\mathbb{X}_r} 	\Phi^* S ) = 0$, thus $\Phi^* S$ is a potential function on $\mathcal{M} \times \mathcal{M}$.
%%%%

Assume now that $\Phi^* S$ is a potential function on $\mathcal{M} \times \mathcal{M}$ so that
\be
(i^{ \mathcal{M} }_D )^* \circ \Phi^* (L_{\mathbb{Z}_l} S)= (i^{ \mathcal{M} }_D )^* ( L_{\mathbb{X}_l} 	\Phi^* S ) =0.
\ee
because it always holds equation \eqref{eqn: lie derivative and pushforward}.
%%%%
Now, we can essentially follow equation \eqref{eqn: S under smooth maps 1} in reverse order  to obtain
\be\label{eqn: S under smooth maps 2}
\begin{split}
0 &=   (i^{ \mathcal{M} }_D )^* \circ \Phi^* (L_{\mathbb{Z}_l} S) =  (\Phi \circ i^{\mathcal{M}}_D)^*(L_{\mathbb{Z}_l} S)  = (i^{\mathcal{N}}_D \circ \phi)^*(L_{\mathbb{Z}_l} S)  =\phi^* \circ (i^{\mathcal{N}}_D)^*(L_{\mathbb{Z}_l} S) .
\end{split}
\ee
%%%%
If $\phi$, and thus $\Phi$, are surjective, equation \eqref{eqn: S under smooth maps 2} implies that  $(i^{\mathcal{N}}_D)^*(L_{\mathbb{Z}_l} S) =0$.
%%%%
A similar proof holds for $(i^{\mathcal{N}}_D)^*(L_{\mathbb{Z}_r} S) $, and we thus conclude that $S$ is a potential function on $\mathcal{N}\times \mathcal{N}$.
%%%%%
\end{proof}

\begin{proposition} \label{prop: g under smooth maps}
Let $\phi$ be a smooth map from $\mathcal{M}$ to $\mathcal{N}$ and $\Phi$ defined as in \eqref{def: Phi}. Let  $S$ be a potential function on $\mathcal{N} \times \mathcal{N}$ so that $\Phi^* S$ is a potential function because of proposition \ref{prop: S under smooth maps}.
%%%%
Let $g_\mathcal{N}$ and $g_{\mathcal{M}}$ be the  tensor fields  extracted from $S$ and $\Phi^{*}S$, respectively, following the procedure outlined in proposition \ref{prop: properties of g}.
%%%%
Then, it holds that  $g_{\mathcal{M}}=\phi^{*}g_{\mathcal{N}}$.
%%%%%
\end{proposition}
\begin{proof}

The proof consists of a direct  check in coordinates.
%%%%
First of all, we introduce a product coordinate chart $(x^{j},y^{k})$ in $\mathcal{M}\times\mathcal{M}$, and a product coordinate chart $(z^{l},w^{r})$ in $\mathcal{N}\times\mathcal{N}$ following what we did right before proposition \ref{prop: properties of lifted fields}.
%%%%%
Then, equation \eqref{eqn: local expression of 2-tensor} implies that
\be\label{eqn: smooth maps and covariant tensors 1}
\begin{split}
g_{\mathcal{M}}&=-\left(\frac{\partial^{2} \Phi^{*}S}{\partial x^{j}\partial y^{k}}\right)_{\mathbf{x}=\mathbf{y}}\mathrm{d}x^{j}\otimes\mathrm{d}x^{k}=\\
%&=-\left(\left(\frac{\partial S}{\partial w^{r}}\right)_{\mathbf{z}=\mathbf{w}=\varphi(\mathbf{x})}\left(\frac{\partial^{2} \Phi^{r}}{\partial x^{j}\partial y^{k}}\right)_{\mathbf{x}=\mathbf{y}} + \left(\frac{\partial^{2} S}{\partial z^{l}\partial w^{r}}\right)_{\mathbf{z}=\mathbf{w}=\varphi(\mathbf{x})}\left(\frac{\partial\Phi^{l}}{\partial x^{j}}\frac{\partial\Phi^{r}}{\partial y^{k}}\right)_{\mathbf{x}=\mathbf{y}}\right)\mathrm{d}x^{j}\otimes\mathrm{d}x^{k}=\\
&=- \left(\frac{\partial^{2} S}{\partial z^{l}\partial w^{r}}\right)_{\mathbf{z}=\mathbf{w}=\varphi(\mathbf{x})} \frac{\partial\varphi^{l}}{\partial x^{j}}\,\frac{\partial\varphi^{r}}{\partial x^{k}} \,\mathrm{d}x^{j}\otimes\mathrm{d}x^{k},
\end{split}
\ee
and
\be\label{eqn: smooth maps and covariant tensors 2}
\begin{split}
g_{\mathcal{N}}&=- \left(\frac{\partial^{2} S}{\partial z^{l}\partial w^{r}}\right)_{\mathbf{z}=\mathbf{w} }    \mathrm{d}z^{l}\otimes\mathrm{d}z^{r}.
\end{split}
\ee
%%%%
Then, the very definition of the pullback of $g_{\mathcal{N}}$ through $\varphi$  implies that
\be
\begin{split}
\varphi^{*}g_{\mathcal{N}}=\left(\left(g_{\mathcal{N}}\right)_{lr}\right)_{\mathbf{z}=\varphi(\mathbf{x})}\frac{\partial\varphi^{l}}{\partial x^{j}}\,\frac{\partial\varphi^{r}}{\partial x^{k}} \,\mathrm{d}x^{j}\otimes\mathrm{d}x^{k},
\end{split}
\ee
which, together with equation \eqref{eqn: smooth maps and covariant tensors 1} and \eqref{eqn: smooth maps and covariant tensors 2}, gives us
\be
g_{\mathcal{M}}=\phi^{*}g_{\mathcal{N}}
\ee
as claimed.

\end{proof}

Proposition \ref{prop: S under smooth maps} and proposition \ref{prop: g under smooth maps} encode  important information on the relation between the symmetry properties of $S$ and those of the associated  tensor field $g$. 
%%%%
Specifically, let $G$ be a Lie group acting on $\mathcal{M}$, let us denote with $\phi_{\mathrm{g}}$ the diffeomorphism of $\mathcal{M}$ induced by $\mathrm{g} \in G$; notice that the map $\Phi_{\mathrm{g}}$ as defined in \eqref{def: Phi} realises an action of $G$ on $\mathcal{M} \times \mathcal{M}$. 
%%%%
If $G$ is a symmetry group of $S$ in the sense that 
\begin{equation}
	S = \Phi_{\mathrm{g}}^* S \; \forall \,\mathrm{g}\in G
\end{equation}
then  it holds 
\be
\phi^*_{\mathrm{g}} g = g \; \forall \,\mathrm{g}\in G,
\ee
that is, $G$ is a symmetry group of the tensor field $g$.
Indeed,  proposition \ref{prop: g under smooth maps} implies that
\be
\phi^*_{\mathrm{g}} g (X,Y)=  i^*_{\mathcal{M}}(L_{_{\mathbb{X}_l}}L_{\mathbb{Y}_l} \Phi^*_{\mathrm{g}} S  )
\ee
which is in turn equal to $g (X,Y)$ for all vector fields $X,Y$ on $\mathcal{M}$ because $G$ is a symmetry group of $S$.
%%%%%%

%%

\begin{proposition}
Let $j=1,2$, $S_{j}$ be a divergence function on $\mathcal{M}_{j}\times\mathcal{M}_{j}$, and $g_{j}$ its associated covariant tensor.
%%%%
Let $\phi\colon\mathcal{M}_{1}\ra\mathcal{M}_{2}$ be a smooth map, and let $\Phi\colon\mathcal{M}_{1}\times\mathcal{M}_{1}\ra\mathcal{M}_{2}\times\mathcal{M}_{2}$ be as in equation \eqref{def: Phi}.
%%%%
If 
\be
S_{1} \geq  \Phi^{*}S_{2},
\ee
then 
\be
g_{1}(X,X)\geq \phi^{*}g_{2}(X,X)
\ee
for all vector fields $X$ on $\mathcal{M}_{1}$.
%%%%
\end{proposition}

\begin{proof}
The function
\be
S_{12}^{\Phi}:=S_{1} - \Phi^{*}S_{2}
\ee
is clearly a divergence function, and thus proposition \ref{prop: divergence functions lead to positive semidefinite tensors} implies that the covariant tensor $g_{12}^{\phi}$ extracted from $S_{12}^{\Phi}$ is positive semidefinite.
%%%
Then, proposition \ref{prop: g under smooth maps} implies that
\be
g_{12}^{\phi} = g_{1 } - \phi^{*}g_{2},
\ee
and the fact that $g_{12}^{\phi}$ is positive semidefinite implies
\be
g_{1 }(X,X)\geq  \phi^{*}g_{2}(X,X)
\ee
for all vector fields $X$ on $\mathcal{M}_{1}$ as claimed.
%%%%

\end{proof}

\section{Unfolding of quantum states, monotone metric tensors, and quantum divergence functions}\label{sec: unfolding of quantum states}

In this section we will introduce and discuss an unfolding procedure for quantum states which will be then  applied to study the family of monotone Riemannian metric tensors recalled in section \ref{sec: Cencov and Petz} and the family of relative $g$-entropies recalled in section \ref{sec: from rel entropies to mono metrics}.
%%%%%

The idea behind the unfolding procedure we will discuss finds its roots in the spectral theorem for self-adjoint linear operators on a Hilbert space.
%%%%
This theorem states that, for any self-adjoint linear operator $\mathbf{a}$ on the finite-dimensional, complex Hilbert space $\hh$, there is an orthonormal basis  $\{ \ket{1}_{a}, \ket{2}_{a},...\ket{n}_{a}\}$ of $\hh$ such that 
\be
\mathbf{a}= \sum_{j} a^{j} \mathbf{e}_{jj}^{a} \,,
\ee
where $\mathbf{e}_{jk}^{a}=\ket{j}_{a}\bra{k}_{a}$, and $\vec{a}=(a^{1},...,a^{n})$ is a vector of real numbers,  the eigenvalues of $\mathbf{a}$, each appearing with its own algebraic multiplicity.
%%%%
In particular, it is not hard to see that the vector of eigenvalues of a quantum state $\rho$ must be a probability vector because $\rho\geq 0$ and $\Tr_{\hh}(\rho)=1$.
%%%%%

Clearly, the representation given above depends on the basis $\{ \ket{1}_{a}, \ket{2}_{a},...\ket{n}_{a}\}$ which in turn depends on the given self-adjoint operator $\mathbf{a}$.
%%%%%
However, since orthonormal bases in $\hh$ are mapped into one another by means of unitary operators, we may always fix an orthonormal basis $\{ \ket{1}, \ket{2},...\ket{n}\}$ and then find a unitary operator $\mathbf{U}_{\mathbf{a}}$ such that
\be\label{eqn: spectral decomposition}
\mathbf{a}= \mathbf{U}_{\mathbf{a}}\left(\sum_{j} a^{j} \,\mathbf{e}_{jj} \right)\mathbf{U}_{\mathbf{a}}^{\dagger},
\ee
with $\mathbf{e}_{jk} = \ket{j}\bra{k} $.
%%%%
Once the basis $\{ \ket{1}, \ket{2},...\ket{n}\}$ is fixed, equation \eqref{eqn: spectral decomposition} allows us to write, in a non unique fashion, any self-adjoint linear operator on $\hh$ in terms of a real n-vector and a unitary operator on $\hh$.
%%%%%
In particular, for a quantum state $\rho$, we obtain the expression
\be\label{eqn: diagonal state}
	\rho = \mathbf{U} \varrho \mathbf{U}^\dagger \qquad \mbox{ with } \qquad \varrho = \sum_{j} p^{j}\,\mathbf{e}_{jj},
\ee 
and since $\vec{p}=(p_{1},...,p_{n})$  must be a probability vector in $\mathbb{R}^{n}$, we may conclude that quantum states can be thought of as a noncommutative counterpart of probability vectors.
%%%%

Equation \eqref{eqn: diagonal state} also tells us that unitary operators and probability vectors are enough to redundantly characterize quantum states, or, stated differently, that  there is a continuous surjective map $\pi$ from the space
\be\label{eqn: unfolding space}
\overline{\mathcal{M}}(\hh):=\SUh\times \overline{\Delta_{n}},
\ee 
where $\overline{\Delta_{n}}$ is the  n-simplex, to the space $\overline{\stsp}(\hh)$ of quantum states on $\hh$  given by
\be\label{eqn: unfolding map}
\pi(\mathbf{U},\vec{p})\,:=\,\mathbf{U}\,\varrho\,\mathbf{U},
\ee
with $\varrho$ as in equation \eqref{eqn: diagonal state}.
%%%%

The existence of the map $\pi$ is what we refer to as an unfolding procedure for quantum states,  in the sense that the space of quantum states is unfolded into the product space $\overline{\mathcal{M}}(\hh)$ whose elements may be thought of a kind of over-complete parametrization of quantum states.
%%%%%

The fact that the unfolding space  $\overline{\mathcal{M}}(\hh)$   contains the  classical probability simplex seems to point out that $\overline{\mathcal{M}}(\hh)$ could be the appropriate space in which to investigate the relation between some classical and quantum properties.
%%%%
Specifically, since we are interested in the Riemannian metric properties characteristic of quantum information geometry, we will investigate what happens to the monotone Riemannian  metric tensors classified by Petz when we pull them back to  a suitable submanifold of  the unfolding space.
%%%

We will see in subsection \ref{subsec: unfolded monotone metric tensors} that every monotone Riemannian metric tensor gets splitted into the sum of a purely classical contribution, which coincides with the classical Fisher-Rao metric tensor, and a quantum contribution that may be written as a weighted sum of the Cartan-Killing metric tensor on $\SUh$ where the weights depend on the classical part (and vanish for a suitable Cartan subalgebra).
%%%%%
We believe that this neat decomposition enforces the idea that monotone Riemannian metric tensors are truly a quantum generalisation of the Fisher-Rao metric tensor.
%%%%

Then, in subsection \ref{subsec: unfolding rel entropies and mono metrics}, we will  pull-back  the relative g-entropies   to be two-point functions on a suitable submanifold of the unfolding space.  Then we will compute the resulting covariant tensor associated with them, by  following the intrinsic reformulation of  what is usually done in information geometry, as described  in section \ref{sec: from 2-functions to covariant tensors}.
%%%%%
We thus obtain  an alternative proof of the results presented in \cite{L-R-1999} relating relative g-entropies with monotone Riemannian metric tensors.
%%%%%

Finally, in subsection \ref{subsec: universal geodesics}, we exploit the unfolded geometry to characterise a family of geodesics which are universal for the family of monotone Riemannian metric tensor in the sense that they are common to all of them.
%%%%
It turns out that these geodesics are precisely the projections of the geodesics of the Fisher-Rao metric tensor through the unfolding map. Since all  monotone Riemannian metric tensors share the same Fisher-Rao component, this explains why these universal geodesics exist.
%%%%%%

Before  proceeding further, we need to briefly investigate the geometry of the unfolding space $\overline{\mathcal{M}}(\hh)$ introduced in equation \eqref{eqn: unfolding space}. 
%%%%
The first thing we note is that $\overline{\mathcal{M}}(\hh)$ is an unfolding space for the whole space of quantum states $\overline{\stsp(\hh)}$, while the natural setting of information geometry is  the smooth manifold $\stsp(\hh)$ of invertible quantum states.
%%%
Therefore, we are led to consider the restriction of $\pi$ to the space
\be\label{eqn: unfolding space 2}
\mathcal{M}(\hh):=\SUh\times\Delta_{n},
\ee
where $\Delta_{n}$ is the open interior of the n-simplex.
%%%%
With an  abuse of notation, we denote the restriction to $\mathcal{M}(\hh)$ of the map $\pi$ introduced in equation \eqref{eqn: unfolding map} again with $\pi$.
%%%%
This map is a smooth surjective map, and   the following proposition characterises the kernel of its tangent map at each point.
%%%%

\begin{proposition}\label{prop: the unfolding map of the invertible density matrices is a surjective differentiable map}
The map $\pi \colon\mathcal{M}(\hh)\rightarrow\stsp(\hh)$ is differentiable, and the kernel of its tangent map at $(\mathbf{U}\,,\vec{p})\in\mathcal{M}_{n}$ is given by $(\imath\mathbf{H}\,,\vec{0})$, where $\mathbf{H}$ is a traceless, self-adjoint operator on $\hh$ $[\mathbf{H}\,,\varrho]=\mathbf{0}$.

\begin{proof*} 
A tangent vector $V_{(\mathbf{U},\vec{p})}$ at $(\mathbf{U},\vec{p})$ can always be written as
\be
V_{(\mathbf{U},\vec{p})}=(\imath\mathbf{H},\vec{a})
\ee
where $\mathbf{H}\in\bh$ is self-adjoint and traceless, and $\vec{a}$ is such that $\sum_{j}a^{j}=0$.
%%%%
Every such tangent vector can be written as the tangent vector at $t=0$ of the   curve $\gamma_{t}$ on $\mathcal{M}(\hh)$ given by
\be
\gamma_{t}(\mathbf{U}\,,\vec{p})=(\mathbf{U}\exp(\imath t\mathbf{H}),\vec{p}_{t})\,
\ee
where $\vec{p}_{t}$ is any curve in the interior of the  simplex $\Delta_{n}$ starting at $\vec{p}_{0}=\vec{p}$ and such that $\left.\frac{\mathrm{d} \vec{p}_{t}}{\mathrm{d}t}\right|_{t=0}=\vec{a}$.
%%%%
Then, the tangent map of $\pi$ at $(\mathbf{U}\,,\vec{p})$ is then
\be
\begin{split}\label{eqn:tangent_of_the_projection_map_pi}
T_{(\mathbf{U},\vec{p})}\pi\left(V_{(\mathbf{U},\vec{p})}\right)&=\frac{\mathrm{d}}{\mathrm{d} t}\left( \pi\left(\gamma(\mathbf{U},\vec{p})\right)\right)_{t=0}= \\
&=\frac{\mathrm{d}}{\mathrm{d}t}\left(\mathbf{U}\exp(\imath t\mathbf{H})\,\rho_{0}(t),\exp(-\imath t\mathbf{H})\mathbf{U}^{\dagger}\right)_{t=0}=\\
&=\mathbf{U}\,\left(\imath\,[\mathbf{H}\,,\varrho] + \sum_{j=1}^{n}a^{j}\mathbf{e}_{jj}\right)\,\mathbf{U}^{\dagger}\,,
\end{split}
\ee
from which it follows that the tangent vector $V_{(\mathbf{U},\vec{p})}$  is sent to the zero tangent vector $\mathbf{0}$ at $\rho=\pi(\mathbf{U},\vec{p})$  if and only if $\vec{a}=\vec{0}$ and $[\mathbf{H}\,,\varrho]=\mathbf{0}$ as claimed.
\end{proof*}

\end{proposition}

From proposition \ref{prop: the unfolding map of the invertible density matrices is a surjective differentiable map} it follows that the fibres of the projection $\pi$ do not share the same dimension, and this indeed reflects the fact that the isotropy subgroup of a given quantum state  with respect to  the action of the unitary group depends on the degeneracy of the eigenvalues of said quantum state.
%%%%
In particular, for  the \emph{maximally mixed state} $\rho_{mm}$, i.e., the state which is proportional to the identity in $\bh$, the isotropy group becomes the whole special unitary group. In particular, every point of $\mathcal{M}(\mathcal{H})$ of the kind $(\mathbf{U},\vec{p}_{1/n})$, with $\vec{p}_{1/n} = (1/n,1/n,\dots,1/n)$ and for all $\mathbf{U} \in \SUh$, has $\rho_{mm}$ as its image via $\pi$. Moreover, the image via the tangent map $T_{(\mathbf{U},\vec{p}_{1/n})} \pi$ of the tangent space $T_{(\mathbf{U},\vec{p}_{1/n})} \mathcal{M}$ has minimal dimension and coincides with the dimension of $\Delta_{n}$, as it is  clear from equation \eqref{eqn:tangent_of_the_projection_map_pi}.
%%%%
In the following, for technical reasons, we will need to restrict our considerations to the open submanifold
\be\label{eqn: open submanifold of unfolded space}
\mathcal{M}_{<}(\hh)\,:=\,\SUh\times\mathscr{C}_{n},
\ee
where 
\be\label{def: Weyl chamber}
 \mathscr{C}_n:= \left\{ \mathbf{p} \in \Delta_n \,\, \big|  \,\, p^{1} <p^{2} < \dots < p^{n} \right\}.
\ee
%%%%
This submanifold is particularly relevant because  the kernel of $T_{(\mathbf{U},\vec{p})}\pi$ at each $(\mathbf{U},\vec{p})\in\mathcal{M}_{<}(\hh)$ is minimal.
%%%%
With an  abuse of notation, we will denote the restriction of $\pi$ to $\mathcal{M}_{<}(\hh)$ again with $\pi$.
%%%%
Clearly, the image of $\mathcal{M}_{<}(\hh)$ through $\pi$ is the open, dense submanifold $\stsp_{md}(\hh)$ of faithful quantum states whose spectrum has the minimal degeneracy, that is, all those quantum states whose eigenvalues are all  distinct.
%%%%

We will now introduce a global differential calculus on the manifold $\mathcal{M}(\hh)$ in terms of a basis of global differential one-forms. 
%%%
This basis will be of capital importance when we will look at monotone quantum metrics and quantum divergences from the unfolded perspective introduced here.
%%%%%
Since $\mathcal{M}_{<}(\hh)$ is an open submanifold of $\mathcal{M}(\hh)$, the same global differential calculus  applies to it.
%%%%%

The main idea is that of introducing a global basis of differential forms and of vector fields  that is adapted to the product structure of $\mathcal{M}(\hh)$.
%%%%
In particular, since $\SUh$ is a Lie group, we naturally have a basis of left-invariant differential one-forms $\{\theta^{j}\}_{j=1,...,n^{2}-1}$ and a basis of globally defined left-invariant vector fields $\{X_{j}\}_{j=1,...,n^{2}-1}$ which is dual to $\{\theta^{j}\}_{j=1,...,n^{2}-1}$.
%%%%%
The Cartesian coordinates $p^{1},...,p^{n}$ of the ambient space $\mathbb{R}^{n}$ define smooth functions on $\Delta_{n}$.
%%%%
Thus, by taking their differentials, we obtain a set of differential one-forms on $\Delta_{n}$, $\{\mathrm{d}p^{j}\}_{j=1,..,n}$ ,  which is an over-complete basis for the module of differential one-forms on $\Delta_{n}$.
%%%
This implies  that the $\mathrm{d}p^{j}$'s  are not functionally independent on $\Delta_{n}$ because
\be\label{eqn: functional dependence of dp}
\sum_{j=1}^{n}\,p^{j}=1\,\Longrightarrow\;\sum_{j=1}^{n}\mathrm{d}p^{j}=0.
\ee
%%%%
Now, we take the pullback of the $\theta^{j}$'s through the canonical projection $\mathrm{pr}_{1}(\mathbf{U},\vec{p})=\mathbf{U}$ (denoted again by $\theta^{j}$), and the pullback of the $\mathrm{d}p^{j}$'s through  the canonical projection $\mathrm{pr}_{2}(\mathbf{U},\vec{p})=\vec{p}$ (denoted again by $\mathrm{d}p^{j}$), so to obtain an over-complete basis $\{\theta^{j},\mathrm{d}p^{k}\}$ of differential one-forms on $\mathcal{M}(\hh)$.
%%%%
Clearly, the basis $\{\theta^{j},\mathrm{d}p^{k}\}$ depends on the choice of left-invariant forms on $\SUh$, and we will now introduce a  particularly relevant choice, that generalise   Pauli matrices $\sigma_{1}, \sigma_{2}$,  $\sigma_{3}$ 
\begin{equation}
	\begin{split}
		 \sigma_1 & = \ket{1}\bra{2} + \ket{2} \bra{1}, \\
		 \sigma_2 & = i \left(\ket{1}\bra{2} - \ket{2} \bra{1}\right)  \\
		 \sigma_3& = \ket{1}\bra{1} - \ket{2} \bra{2}   
	\end{split}
\end{equation}
 to higher dimensional cases.
%%%% 
To this end, we define the family $\{i \tau^1_{jk}, i\tau^2_{jk}, i\tau^3_{lm} \}$ of operators on $\hh$,  with $ 1 \le j < k \le n$  and $ 1 \le l = m-1 \le n-1 $, setting
\begin{equation}\label{eqn: basis generalized Pauli}
	\begin{split}
		\tau^1_{jk} & = \ket{j}\bra{k} + \ket{k} \bra{j}, \\
		\tau^2_{jk} & = i \left(\ket{j}\bra{k} - \ket{k} \bra{j}\right), \\
		\tau^3_{lm} & = \ket{l}\bra{l} - \ket{m} \bra{m}.
	\end{split}
\end{equation}
%%%
It is not hard to show that $\{i \tau^1_{jk}, i\tau^2_{jk}, i\tau^3_{lm} \}$ provides a basis of the Lie algebra $\mathfrak{su} (\mathcal{H})$ of $\SUh$.
%%%%
Moreover, when $\mathrm{dim}(\hh)=2$, it is readily seen that $\tau_{12}^{1}=\sigma_{1}$, $\tau_{12}^{2}=\sigma_{2}$, and $\tau_{12}^{3}=\sigma_{3}$, thus clarifying the sense in which the family $\{ \tau^1_{jk}, \tau^2_{jk}, \tau^3_{lm} \}$ is a generalization of Pauli matrices.
%%%%%
Given the basis $\{ i\tau^1_{jk}, i\tau^2_{jk}, i\tau^3_{lm} \}$ of $\mathfrak{su} (\mathcal{H})$, we immediately obtain a basis $\{ \theta^1_{jk}, \theta^2_{jk}, \theta^3_{lm} \} $ of left-invariant 1-forms on $\SUh$ by means of the relation
\begin{equation} \label{eqn: decomposition Cartan form}
	U^\dagger d U =    \sum_{\substack{j = 1 \\ k > j}}^{n} i\tau^1_{jk} \, \theta_1^{jk} + i\tau^2_{jk} \, \theta_2^{jk} + \sum_{\substack{l=1 \\ m = l + 1}}^{n-1} i\tau^3_{lm} \, \theta_3^{lm},
\end{equation}
where $U^\dagger dU$ is the Cartan 1-form of $\SUh$.
%%%%%
Finally, a basis of left-invariant vector fields on $\SUh$ is obtained by considering the dual basis of  $\{ \theta^1_{jk}, \theta^2_{jk}, \theta^3_{lm} \} $.
%%%%

\subsection{Unfolding of monotone metrics}\label{subsec: unfolded monotone metric tensors}

One of the main objectives of this work is to investigate how the unfolding perspective introduced before along the lines of \cite{C-DC-L-M-M-V-V-2018, M-M-V-V-2017} modifies our understanding of the geometrical structures of quantum information geometry.
%%%%
In particular, here we will see what happens to the quantum monotone metric tensor $\GG^{\hh}_{f}$ reviewed in section \ref{sec: Cencov and Petz} when we pull it back to the unfolding space $\mathcal{M}(\hh)=\SUh\times\Delta_{n}$, with $\mathrm{dim}_{\mathbb{C}}(\hh)=n$.
%%%%
The result will be a decomposition of the pullback into the sum of two contributions, one which is purely classical in the sense that it coincides with the Fisher-Rao metric tensor on $\Delta_{n}$, and the other which is exquisitely non-classical.
%%%%%

Recall that a tangent vector at $(\mathbf{U},\vec{p})\in\SUh\times\Delta_{n}$ can be written as $(   i\mathbf{H} ,\vec{a})$ where $\mathbf{H}\in\mathcal{B}_{sa}(\hh)$ and $\sum_{j}a_{j}=0$.
%%%
Then, the action of the tangent map $T_{(\mathbf{U},\vec{p})}\pi$ on  the tangent vector  $(   i\mathbf{H} ,\vec{a})$ reads
\be\label{eqn: pushforward of tangent vectors through pi}
T_{(\mathbf{U},\vec{p})}\pi(i\mathbf{H} ,\vec{a})\,=\,\Phi_{\mathbf{U}}\left(i[\mathbf{H},\varrho ] + \sum_{j=1}^{n}a^{j}\mathbf{e}_{jj} \right)\equiv\Phi_{\mathbf{U}}(V_{\mathbf{H}}^{\vec{a}}).
\ee
%%%
Therefore, by using equation \eqref{eqn: Petz metric 1},  we obtain 
\begin{eqnarray}\label{eqn: pullback monotone metric 1}
\left(\pi^{*}\GG^{\hh}_{f}\right)_{(\mathbf{U},\vec{p})}\left((i\mathbf{H} ,\vec{a}),(i\mathbf{K} ,\vec{b})\right)&= &\left(\GG^{\hh}_{f}\right)_{\pi(\mathbf{U},\vec{p})}\left(T_{(\mathbf{U},\vec{p})}\pi(i\mathbf{H} ,\vec{a}),T_{(\mathbf{U},\vec{p})}\pi(i\mathbf{K} ,\vec{b})\right) \nonumber \\
%&\stackrel{\mbox{\eqref{eqn: pushforward of tangent vectors through pi}}}
&=&\left(\GG^{\hh}_{f}\right)_{\pi(\mathbf{U},\vec{p})}\left(\Phi_{\mathbf{U}}(V_{\mathbf{H}}^{\vec{a}}),\Phi_{\mathbf{U}}(V_{\mathbf{K}}^{\vec{b}})\right) \nonumber\\
%&\stackrel{\mbox{\eqref{eqn: Petz metric 1}}}
&=& \langle \Phi_{\mathbf{U}}(V_{\mathbf{H}}^{\vec{a}}),T_{\Phi_{\mathbf{U}}(\varrho )}^{f}\,\Phi_{\mathbf{U}}(V_{\mathbf{K}}^{\vec{b}}) \rangle_{\hh}^{HS} \nonumber \\
& =&\langle V_{\mathbf{H}}^{\vec{a}},\Phi_{\mathbf{U}}^{\dagger}T_{\Phi_{\mathbf{U}}(\varrho )}^{f}\,\Phi_{\mathbf{U}}(V_{\mathbf{K}}^{\vec{b}}) \rangle_{\hh}^{HS}  =\langle V_{\mathbf{H}}^{\vec{a}}, T_{ \varrho  }^{f} (V_{\mathbf{K}}^{\vec{b}}) \rangle_{\hh}^{HS}
\end{eqnarray}
where we used unitary invariance in the last step.
%%%
%%%%%
Recalling equation \eqref{eqn: eigendecomposition of Petz superoperator}, we have
\be\label{eqn: pullback monotone metric 2}
\begin{split}
T_{ \varrho  }^{f}\, (V_{\mathbf{K}}^{\vec{b}})&=\sum_{j,k=1}^{n}\,\left(p^{k} \,f\left(\frac{p^{j} }{p^{k} }\right)\right)^{-1}\,E_{kj} \left(i[\mathbf{K},\varrho ] + \sum_{l=1}^{n}a^{l}\mathbf{e}_{ll} \right) \\
&=\sum_{j,k=1}^{n}\,i\,K^{jk} \left(p^{k} -p^{j}  \right)\,\left(p^{k} \,f\left(\frac{p^{j} }{p^{k} }\right)\right)^{-1}\,\mathbf{e}_{jk}\,+\,\sum_{j=1}^{n}\,\frac{ b^{j}}{p^{j} }\,\mathbf{e}_{jj} , 
\end{split}
\ee
where $K^{jk} $ is the $(j,k)$-th matrix element of $\mathbf{K}$ with respect to the basis of eigenvectors of $\varrho$.
%%%%
Inserting equation \eqref{eqn: pullback monotone metric 2} into equation \eqref{eqn: pullback monotone metric 1}  we finally obtain
\be\label{eqn: unfolded Petz metric 1}
\begin{split}
\left(\pi^{*}\GG^{\hh}_{f}\right)_{(\mathbf{U},\vec{p})}\left((i\mathbf{H} ,\vec{a}),(i\mathbf{K} ,\vec{b})\right)&=\sum_{j,k=1}^{n}\, \,H^{kj}  K^{jk} \left(p^{k} -p^{j}  \right)^{2}\,\left(p^{k} \,f\left(\frac{p^{j} }{p^{k} }\right)\right)^{-1} \,+\,\sum_{j=1}^{n}\,\frac{a^{j} b^{j}}{p^{j} }  .
\end{split}
\ee
%%%%
Note that the first term in the right-hand-side of the last expression 
%in equation \eqref{eqn: unfolded Petz metric 1}  
vanishes   when $p^{j}=p^{k} $, and thus, in particular, it does not get contribution from the diagonal part of $\rho=\mathbf{U}\varrho \mathbf{U}^{\dagger}$.
%%%%
%%%%
Loosely speaking, we may say that the first term in the right-hand-side of the last expression in equation \eqref{eqn: unfolded Petz metric 1} is  a purely non-classical contribution, while the classical contribution is encoded in   the second term.
%in the right-hand-side of the last expression in equation \eqref{eqn: unfolded Petz metric 1} .
%%%%%
Indeed, we see that when $[\mathbf{H},\varrho ]=0$ (or $[\mathbf{K},\varrho ]=0$) we obtain an expression which coincides with the purely classical contribution given by the Fisher-Rao metric tensor
\be\label{eqn: unfolded Petz and FR}
\left(\pi^{*}\GG^{\hh}_{f}\right)_{(\mathbf{U},\vec{p})}\left((i\mathbf{H} ,\vec{a}),(i\mathbf{K} ,\vec{b})\right)=(\GG_{FR})_{\vec{p}}\,(\vec{a},\vec{b}).
\ee
%%%%
This result is not surprising in light of equation \eqref{eqn: Petz and F-R}.
%%%%
However, the unfolded perspective discussed here allows for a global decomposition of the unfolded metric tensor into the sum of a classical part and a non-classical one, while equation \eqref{eqn: Petz and F-R} expresses only a pointwise relation.
%%%%

Let us elaborate a little bit more on this instance.
%%%%%
The unfolding space $\mathcal{M}(\hh)=\SUh\times \Delta_{n}$ is a product manifold and thus   every vector field $X$ on it may always be decomposed as
\be\label{eqn: decomposition of vector fields on unfolding manifold}
X=f_{1}X_{1} + f_{2}X_{2},
\ee
where $X_{1}$ is a vector field which is   tangent to $\SUh$, $X_{2}$ is a vector field which is tangent to $\Delta_{n}$, and $f_{1}$ and $f_{2}$ are arbitrary smooth functions on the whole $\mathcal{M}(\hh)$.
%%%%
Then,  equation \eqref{eqn: unfolded Petz metric 1} amounts to the fact that the pullback  to $\mathcal{M}(\hh)$ of every monotone quantum metric $\GG_{f}^{\hh}$ decomposes into the sum
\be\label{eqn: unfolded Petz metric 2}
\pi^{*}\GG_{f}^{\hh}\,=\,\GG_{f}^{\SUh} + \GG_{f}^{\Delta_{n}},
\ee
where $\GG_{f}^{\SUh}(hX,Y)=0$ for every smooth function $h$ on $\mathcal{M}(\hh)$ and every vector field  $Y$ on $\mathcal{M}(\hh)$, whenever $X$ is tangent to $\Delta_{n}$. Analogously, $\GG_{f}^{\Delta_{n}}(hX,Y)=0$ for every smooth function $h$ on $\mathcal{M}(\hh)$ and every vector field  $Y$ on $\mathcal{M}(\hh)$ whenever $X$ is tangent to $\SUh$.
%%%%
Then, equation \eqref{eqn: unfolded Petz and FR} implies
\be\label{eqn: unfolded Petz metric 3}
\GG_{f}^{\Delta_{n}}=\pi^{*}\GG_{FR}^{n},
\ee
where $\GG_{FR}$ is the Fisher-Rao metric tensor on $\Delta_{n}$.
%%%%%
Concerning $\GG_{f}^{\SUh}$, we first note that, on using the basis $\{ \tau^1_{jk}, \tau^2_{jk}, \tau^3_{lm} \}$ of traceless, self-adjoint operators introduced in equation \eqref{eqn: basis generalized Pauli}, it is not hard to prove that
\be
\begin{split}
H^{kj}&= H_{1}^{jk} + \imath H_{2}^{jk} \quad \mbox{ for } j<k \\
H^{jk}&= H_{1}^{jk} - \imath H_{2}^{jk} \quad \mbox{ for } j<k ,
\end{split}
\ee
where $H_{1}^{jk}$ is the component of $\mathbf{H}$ along $\tau_{jk}^{1}$ and $H_{2}^{jk}$ is the component of $\mathbf{H}$ along $\tau_{jk}^{2}$, so that%recalling the first term in the left hand side of 
 from equation \eqref{eqn: unfolded Petz metric 1}, we can write 
\be\label{eqn: unfolded Petz metric 4}
\begin{split}
\left(\GG_{f}^{\SUh}\right)_{(\mathbf{U},\vec{p})}\left((i\mathbf{H} ,\vec{a}),(i\mathbf{K} ,\vec{b})\right)&=\sum_{k>j=1}^{n} \,C^{jk}\,\left(p^{k} -p^{j}  \right)^{2}\,\left(p^{k} \,f\left(\frac{p^{j} }{p^{k} }\right)\right)^{-1}    
\end{split}
\ee
with
\be
C^{jk}=2\left(H^{jk}_{1} K^{jk}_{1} + H^{jk}_{2} K^{jk}_{2}\right).
\ee
%%%%
Consequently,  on using  the  one-forms introduced in equation \eqref{eqn: decomposition Cartan form}, equation \eqref{eqn: unfolded Petz metric 4} yields 
\be\label{eqn: unfolded Petz metric 5}
\GG_{f}^{\SUh}=\sum_{k>j=1}^{n} \,2\left(p^{k} -p^{j}  \right)^{2}\,\left(p^{k} \,f\left(\frac{p^{j} }{p^{k} }\right)\right)^{-1} \left(\theta^{jk}_{1}\otimes\theta^{jk}_{1} + \theta^{jk}_{2}\otimes \theta^{jk}_{2}\right).
\ee
%%%%
Putting together equation \eqref{eqn: unfolded Petz metric 2}, equation \eqref{eqn: unfolded Petz metric 3}, and equation  \eqref{eqn: unfolded Petz metric 5}, we finally obtain
\be\label{eqn: unfolded Petz metric 6}
\pi^{*}\GG_{f}^{\hh}=\sum_{k>j=1}^{n} \,2\left(p^{k} -p^{j}  \right)^{2}\,\left(p^{k} \,f\left(\frac{p^{j} }{p^{k} }\right)\right)^{-1} \left(\theta^{jk}_{1}\otimes\theta^{jk}_{1} + \theta^{jk}_{2}\otimes \theta^{jk}_{2}\right) + \pi^{*}\GG_{FR}^{n}
\ee
%%%%
We thus see that the unfolded perspective allows us to globally separate the classical and non-classical contributions to $\pi^{*}\GG_{f}^{\hh}$ as claimed.
%%%%

%%%%%

\subsection{Unfolding of relative entropies and monotone metrics} 
\label{subsec: unfolding rel entropies and mono metrics}

In this section, we will  extract the covariant tensor associated with the relative g-entropy introduced in equation \eqref{def: entropy Lev Ruskai}.
%%%%%
To this, we shall exploit the extraction algorithm introduced in section \ref{sec: from 2-functions to covariant tensors}, and we will see that the result is essentially the covariant tensor found in equation \eqref{eqn: unfolded Petz metric 6} by pulling back the monotone metric tensor $\GG_{f}^{\hh}$ to the unfolding space $\mathcal{M}(\hh)$.
%%%%%
Because of propositions \ref{prop: S under smooth maps} and  \ref{prop: g under smooth maps}, this result represents an unfolded version of the result presented in \cite{L-R-1999}.
%%%%%

Let us start by setting
\begin{gather}
	\rho = \pi(U, \mathbf{p}) = \sum_{j} p^j U \mathbf{e}_{jj} U^\dagger,\\
	\sigma = \pi(V, \mathbf{q}) = \sum_{j} q^j V \mathbf{e}_{jj} V^\dagger,
\end{gather}
so that the action of the super-operators $L_\sigma$ and $R_\sigma^{-1}$ can be written as
\begin{gather}\label{eqn: left and right superoperators}
	L_\sigma (A) = \sum_{j} q^j V \mathbf{e}_{jj} V^\dagger A,\\
	R_\rho^{-1} (A) = \sum_{j} (p^j)^{-1} A U \mathbf{e}_{jj} U^\dagger.
\end{gather}
%%%%%
Then, we define the projection  
%%%%%
\begin{equation}
	\Pi: \mathcal{M(H)} \times \mathcal{M(H)} \ni \big((U,\mathbf{p}),(V,\mathbf{q})\big) \mapsto \big(\pi(U,\mathbf{p}),\pi(V,\mathbf{q})\big) \in \stsp_n \times \stsp_n,
\end{equation} 
so that we can pull-back to $\mathcal{M(H)}\times\mathcal{M(H)}$ the relative g-entropy $S_{g}^{\hh}$ given in equation \eqref{def: entropy Lev Ruskai}.
%%%%
To this purpose, we first note that, because of equation \eqref{eqn: left and right superoperators}, it holds
\be
L_{\sigma}R_{\rho}^{-1}(\mathbf{e}_{jk}^{\sigma\rho})=\frac{q^{j}}{p^{k}} \mathbf{e}_{jk}^{\sigma\rho},
\ee
where
\be
\mathbf{e}_{jk}^{\sigma\rho}=V\,\mathbf{e}_{jk}\,U^{\dagger}.
\ee
%%% 
It follows that $L_{\sigma}R_{\rho}^{-1}$ can be diagonalized according to\footnote{Here, $\mathbf{e}_{jk}^{\sigma\rho}$ should be thought of as an element of $\bh$ endowed with the Hilbert-Schmidt inner product, and this explains the bra and ket notation.}
\be
L_{\sigma}R_{\rho}^{-1}=\sum_{j,k=1}^{n}\frac{q^{j}}{p^{k}} |\mathbf{e}_{jk}^{\sigma\rho}\rangle\langle   \mathbf{e}_{jk}^{\sigma\rho}|
\ee
and thus
\be
g\left(L_{\pi(V,\mathbf{\sigma})}R_{\pi(U,\mathbf{p})}^{-1}\right)= g\left(L_{\sigma}R_{\rho}^{-1}\right)=\sum_{j,k=1}^{n}\,g\left(\frac{q^{j}}{p^{k}}\right)\, |\mathbf{e}_{jk}^{\sigma\rho}\rangle\langle   \mathbf{e}_{jk}^{\sigma\rho}| .
\ee
%%%%
It is then a matter of straightforward computation to show that the pull-back to $\mathcal{M}(\hh)\times\mathcal{M}(\hh)$ through $\Pi$, of the relative g-entropy $S_{g}^{\hh}$ given in equation \eqref{def: entropy Lev Ruskai},  reads
\be\label{eqn: pull-back of g-entropy}
\begin{split}
	\Pi^*\left(S_{g}^{\hh}\right)((U,\mathbf{p}),(V,\mathbf{q})) &=   \sum_{j,k}  g\left( \frac{q^j}{p^k} \right)  p^k    \bra{k} U^\dagger V \ket{j} \bra{j} V^\dagger U  \ket{k}.
	\end{split}
\ee

Let us check that $\Pi^*\left(S_{g}^{\hh}\right)$ is a potential function according to definition \ref{def: potential functions}, namely that  it satisfies equation \eqref{eqn: Lie der w.r.t. Xl of S} and equation \eqref{eqn: Lie der w.r.t. Xr of S}.
%%%%
On computing $L_{\mathbb{Y}_r} \left( \pi^*S_{g}^{\hh}\right)$ by means of  the equality
\be 
L_{\mathbb{Y}_r} \left( \Pi^*S_{g}^{\hh}  \right) = i_{\mathbb{Y}_r}  d \left( \Pi^*S_{g}^{\hh}   \right)
\ee
following from Cartan's magic formula, and by 
using proposition \ref{prop: properties of lifted fields} and 
 equation \eqref{eqn: decomposition of vector fields on unfolding manifold}, we get
\be \label{eqn: unfolded tensor from relative g-entropy 1}
	\begin{split}
		 i_{\mathbb{Y}_r}  d \left( \Pi^*S_{g}^{\hh}   \right) & = \sum_{j,k}g'\left( \frac{q^j}{p^k}  \right)   dq^j (\mathbb{Y}_r) \bra{k} U^\dagger V \ket{j} \bra{j} V^\dagger U  \ket{k}  \\
		& + \sum_{j, k} g\left( \frac{q^j}{p^k}  \right) p^k \bra{k} U^\dagger dV (\mathbb{Y}_r) \ket{j} \bra{j} V^\dagger U  \ket{k} \\
		& -  \sum_{j, k} g\left( \frac{q^j}{p^k}  \right) p^k \bra{k} U^\dagger V \ket{j} \bra{j} V^\dagger dV (\mathbb{Y}_r) \, V^\dagger U  \ket{k} .
	\end{split}
\ee
The pull-back with respect to  the immersion $i_D$ is equivalent to set  $U=V$ and $\mathbf{p}=\mathbf{q}$, and thus
\be \label{eqn: unfolded tensor from relative g-entropy 2}
	\begin{split}
		i_D^{*}\left( L_{\mathbb{Y}_r} \left( \Pi^*S_{g}^{\hh}  \right)\right) & = g'\left(1  \right)\sum_{j=1}^{n}\,   dp^j (Y)=0
	\end{split}
\ee
because  of equation \eqref{eqn: functional dependence of dp}.
%%%%
Proceeding analogously, we get 
\be \label{eqn: unfolded tensor from relative g-entropy 3}
	\begin{split}
		 i_D^{*}\left(L_{\mathbb{Y}_r} \left( \Pi^*S_{g}^{\hh}  \right)\right) & = - g'\left(1  \right)\sum_{j=1}^{n} dp^j (Y)=0,
	\end{split}
\ee
and we conclude that $\Pi^*S_{g}^{\hh}$ is a potential function according to definition \ref{def: potential functions}.
%%%%
Moreover, since $\Pi$ is surjective, proposition \ref{prop: S under smooth maps} implies that $S_{g}^{\hh}$ is a potential function.
%%%%%

Now, following the procedure described in section \ref{sec: from 2-functions to covariant tensors}, we obtain a symmetric covariant tensor field on $\mathcal{M}(\hh)$ from the two-point function $\Pi^* S^{\hh}_{g}$ by means of 
\be\label{eqn: unfolded tensor from relative g-entropy 4}
	\GG^{\mathcal{M(H)}}_{g}(X,Y) = - i_D^*\left(L_{\mathbb{X}_l} L_{\mathbb{Y}_r} \left( \Pi^* S^{\hh}_{g} \right) \right),
\ee
where $X$ and $Y$ are arbitrary vector fields on $\mathcal{M(H)}$, and $\mathbb{X}_l, \mathbb{Y}_r$ are their left and right lift as in definition \ref{def: l and r lifts}. 
%%%%%
For this purpose, we start considering the first term in the right hand side of \eqref{eqn: unfolded tensor from relative g-entropy 1} given by
\be
A:=\sum_{j,k} g'\left( \frac{q^j}{p^k}  \right)  dq^j (\mathbb{Y}_r)  \bra{k} U^\dagger V \ket{j} \bra{j} V^\dagger U  \ket{k}.
\ee
%%%%
We have 
\be\label{eqn: unfolded tensor from relative g-entropy 5}
\begin{split}
i_{D}^{*}\left(\mathcal{L}_{\mathbb{X}_{l}}\left(A\right)\right)&= -\sum_{j }    \frac{g''(1)}{p^{j}} \,\mathrm{d}p^{j}(X)  \,\mathrm{d}p^j (Y)=-g''(1)\,\pi^{*}\left(\GG_{FR}^{n}\right)(X,Y)  , 
\end{split}
\ee
where $\GG_{FR}^{n}$ is the Fisher-Rao metric tensor recalled in equation \eqref{eqn: Fisher-Rao metric tensor}, and where we used the fact that  $\mathrm{d}U^{\dagger}=-U^{\dagger}\mathrm{d}U\,U^{\dagger}$ to eliminate the contributions coming from the unitary group.
%%%%%
Then, we pass to the second term in the right hand side of \eqref{eqn: unfolded tensor from relative g-entropy 1} given by
\be
B:=\sum_{j, k} g\left( \frac{q^j}{p^k}  \right) p^k \bra{k} U^\dagger dV (\mathbb{Y}_r) \ket{j} \bra{j} V^\dagger U  \ket{k} .
\ee
%%%%
In this case, we have
\be\label{eqn: unfolded tensor from relative g-entropy 6}
\begin{split}
i_{D}^{*}\left(\mathcal{L}_{\mathbb{X}_{l}}\left(B\right)\right)&=  -g'(1)\sum_{j=1}^{n}\,\mathrm{d}p^{j}(X) \bra{j} U^\dagger \mathrm{d}U (Y) \ket{j} \,+  \\
 & - g\left( 1\right)\sum_{j=1}^{k}  p^j \bra{j} U^\dagger \mathrm{d}U (Y) U^{\dagger}\mathrm{d}U(X) \ket{j}   + \\
 &+ g\left(1  \right) \sum_{j=1}^{n} p^j \,  \bra{j} U^\dagger \mathrm{d}U(Y)  \ket{j} \bra{j} U^\dagger \mathrm{d}U (X) \ket{j} +\\
& +  \sum_{k> j}   g\left( \frac{p^j}{p^k}  \right) p^k\,\bra{k} U^\dagger \mathrm{d}U (Y) \ket{j} \bra{j} U^\dagger \mathrm{d}U(X)  \ket{k}  + \\
& +  \sum_{k> j} g\left( \frac{p^k}{p^j}  \right) p^j \, \bra{j} U^\dagger \mathrm{d}U (Y) \ket{k} \bra{k} U^\dagger \mathrm{d}U(X)  \ket{j}  .
\end{split}
\ee
%%%%
%%%%
The third term in the right hand side of \eqref{eqn: unfolded tensor from relative g-entropy 1} given by
\be
C:= -  \sum_{j , k} g\left( \frac{q^j}{p^k}  \right) p^k \bra{k} U^\dagger V \ket{j} \bra{j} V^\dagger dV (\mathbb{Y}_r) \, V^\dagger U  \ket{k}
\ee
is computed in analogy with equation \eqref{eqn: unfolded tensor from relative g-entropy 4} leading to 
\be\label{eqn: unfolded tensor from relative g-entropy 7}
\begin{split}
i_{D}^{*}\left(\mathcal{L}_{\mathbb{X}_{l}}\left(C\right)\right)&=   g'(1)\sum_{j=1}^{n}\,\mathrm{d}p^{j}(X) \bra{j} U^\dagger \mathrm{d}U (Y) \ket{j} \,+ \\
& - g\left( 1  \right) \sum_{j =1}^{n}  p^j   \bra{j} U^\dagger \mathrm{d}U (X) \, U^\dagger \mathrm{d}U(Y)  \ket{j} +\\
& + g\left( 1\right)\sum_{j =1}^{n}  p^j \bra{j} U^\dagger \mathrm{d}U(X) \ket{j} \bra{j} U^\dagger \mathrm{d}U (Y) \,   \ket{j}\\
& +  \sum_{k> j}   g\left( \frac{p^j}{p^k}  \right) p^k\, \bra{k} U^\dagger \mathrm{d}U(X)  \ket{j} \bra{j} U^\dagger \mathrm{d}U (Y) \ket{k} + \\
& +  \sum_{k> j} g\left( \frac{p^k}{p^j}  \right) p^j \,  \bra{j} U^\dagger \mathrm{d}U(X)  \ket{k} \bra{k} U^\dagger \mathrm{d}U (Y) \ket{j} .
\end{split}
\ee
%%%%
Collecting the results, equation \eqref{eqn: unfolded tensor from relative g-entropy 4} becomes
\be\label{eqn: expression metric on the unfolding space 1}
	\begin{split}
		 \GG^{ \mathcal{M(H)} }_{ g }(X,Y)   & = g''(1)\,\pi^{*} \GG_{FR}^{n} (X,Y)   \\
		 &+ g\left( 1\right)\sum_{j=1}^{k}  p^j \bra{j} U^\dagger \mathrm{d}U (X) U^{\dagger}\mathrm{d}U(Y) + U^\dagger \mathrm{d}U (Y) \, U^\dagger \mathrm{d}U(X)  \ket{j}  \\
& - 2g\left(1  \right) \sum_{j=1}^{n} p^j \,  \bra{j} U^\dagger \mathrm{d}U(X)  \ket{j} \bra{j} U^\dagger \mathrm{d}U (Y) \ket{j}  \\
& -2 \sum_{k>j} \left(g\left( \frac{p^j}{p^k}  \right) p^k + g\left( \frac{p^k}{p^j}  \right) p^j \right)  \,\,\Re\left(\bra{k} U^\dagger \mathrm{d}U (X) \ket{j} \bra{j} U^\dagger \mathrm{d}U(Y)  \ket{k}\right)  .
	\end{split}
\ee
In terms of  the basis introduced in \eqref{eqn: basis generalized Pauli}, we may use equation  \eqref{eqn: decomposition Cartan form}  to obtain
\be\label{eqn: expression metric on the unfolding space 2}
\begin{split}
\GG^{ \mathcal{M(H)} }_{ g }    & = g''(1)\,\pi^{*} \GG_{FR}^{n}   
 - 2g\left( 1\right)\sum_{k>j,m>l}\Tr(\varrho\{\tau_{jk}^{1},\tau_{lm}^{1}\})\theta^{jk}_{1}\otimes\theta^{lm}_{1}   \\
&- 2g\left( 1\right)\sum_{k>j,m>l}\Tr(\varrho\{\tau_{jk}^{1},\tau_{lm}^{2}\})\theta^{jk}_{1}\otimes\theta^{lm}_{2}    
- 2g\left( 1\right)\sum_{k>j,m-1=l}\Tr(\varrho\{\tau_{jk}^{1},\tau_{lm}^{3}\})\theta^{jk}_{1}\otimes\theta^{lm}_{3}    \\
&- 2g\left( 1\right)\sum_{k>j,m>l}\Tr(\varrho\{\tau_{jk}^{2},\tau_{lm}^{2}\})\theta^{jk}_{2}\otimes\theta^{lm}_{2}    
- 2g\left( 1\right)\sum_{k>j,m-1=l}\Tr(\varrho\{\tau_{jk}^{2},\tau_{lm}^{3}\})\theta^{jk}_{2}\otimes\theta^{lm}_{3}   \\
&- 2g\left( 1\right)\sum_{k-1=j,m-1=l}\Tr(\varrho\{\tau_{jk}^{3},\tau_{lm}^{3}\})\theta^{jk}_{3}\otimes\theta^{lm}_{3}   
 + 2g\left(1  \right)\sum_{m-1=l=1}^{n-1}(p^{l} + p^{m})\theta^{lm}_{3}\otimes\theta^{lm}_{3} \\
& +2 \sum_{k>j} \left(g\left( \frac{p^j}{p^k}  \right) p^k + g\left( \frac{p^k}{p^j}  \right) p^j \right)  \left( \theta_1^{jk} \otimes \theta_1^{jk} + \theta_2^{jk} \otimes \theta_2^{jk} \right),
\end{split}
\ee
where the Jordan product introduced in \eqref{eqn: Jordan product} as been employed.
%%%%

Now, we want to compare $\GG^{ \mathcal{M(H)} }_{ g } $ with the pullback to $\mathcal{M}(\hh)$ of the monotone metric tensor $\GG_{f}^{\hh}$  given in equation \eqref{eqn: unfolded Petz metric 6}, to see if they agree for some choice of the function $g$.
%%%%%
To this, we first note that $\GG^{ \mathcal{M(H)} }_{ g } $ contains terms in $\theta^{jk}_{3}\otimes\theta^{lm}_{3}$, while $\Pi^{*}\GG_{f}^{\hh}$ does not.
%%%%
Therefore, the first condition we must impose on $g$ in order for $\GG^{ \mathcal{M(H)} }_{ g } $ to agree with $\Pi^{*}\GG_{f}^{\hh}$ is
\be
g(1)=0.
\ee
%%%%
With this condition, equation \eqref{eqn: expression metric on the unfolding space 2} simplifies to
\be\label{eqn: expression metric on the unfolding space 3}
\begin{split}
\GG^{ \mathcal{M(H)} }_{ g }    & = g''(1)\,\pi^{*} \GG_{FR}^{n}    +2 \sum_{k>j} \left(g\left( \frac{p^j}{p^k}  \right) p^k + g\left( \frac{p^k}{p^j}  \right) p^j \right)  \,\left( \theta_1^{jk} \otimes \theta_1^{jk} + \theta_2^{jk} \otimes \theta_2^{jk} \right).
\end{split}
\ee
%%%
If we focus only on the “unitary part” of $\GG^{ \mathcal{M(H)} }_{ g }$ and $\Pi^{*}\GG_{f}^{\hh}$, comparing equation \eqref{eqn: expression metric on the unfolding space 3} with equation \eqref{eqn: unfolded Petz metric 6}, we immediately see that they are equal if and only if
\be\label{eqn: f and g}
f\left(x\right) = \frac{\left(1 - x  \right)^{2}}{g\left( x\right) + x g\left(x^{-1}\right)} .
\ee
%%%%
In this case, a direct computation shows that
\be
g''(1)=f(1){=}1,
\ee
where equation \eqref{eqn: Petz function} has been used.
Thus we conclude that
\be
\GG^{\mathcal{M(H)}}_g =\pi^{*}\GG^{\hh}_{f}
\ee
whenever   $g(1)=0$ and  equation \eqref{eqn: f and g} holds.
%%%%
 
\subsection{Universal geodesics}\label{subsec: universal geodesics}

In this subsection we will see that all the monotone Riemannian metric tensors classified by Petz share a common family of geodesics which are, essentially, the projections of the geodesics of the Fisher-Rao metric tensor through the unfolding map $\pi$.
%%%%%

For technical reasons, we need to focus our attention to the open  submanifold $\mathcal{M}_{<}(\hh)$ of $\mathcal{M}(\hh)$ introduced in equation \eqref{eqn: open submanifold of unfolded space}.
%%%%
Indeed, the rank of the covariant tensor $\pi^{*}\GG_{f}^{\hh}$  given in equation \eqref{eqn: unfolded Petz metric 6} is maximal on this submanifold.
%%%%
Then, we define the covariant tensor $\mathcal{G}_{f}$ on  $\mathcal{M}_{<}(\hh)$ by
\be
\mathcal{G}_{f} := \pi^{*}\GG_{f}^{\hh} +  \sum\,\theta^{lm}_{3}\otimes\theta^{lm}_{3}.
\ee
%%%%%
Clearly, $\mathcal{G}_{f}$ is a Riemannian metric tensor on $\mathcal{M}_{<}(\hh)$, and it is a matter of direct inspection to see that $\pi\colon\mathcal{M}_{<}(\hh)\ra\stsp_{md}(\hh)$ is a Riemannian submersion between $(\mathcal{M}_{<}(\hh),\mathcal{G}_{f})$ and $(\stsp_{md}(\hh),\GG_{f}^{\hh})$.
%%%%%
This instance allows us to exploit the power of the theory of Riemannian submersions in order to study the Riemannian geometry of  $(\stsp_{md}(\hh),\GG_{f}^{\hh})$ in terms of the Riemannian geometry of $(\mathcal{M}_{<}(\hh),\mathcal{G}_{f})$ which is easier to handle.
%%%%

In particular, we are interested in the study of the geodesics of $(\stsp_{md}(\hh),\GG_{f}^{\hh})$ .
%%%
According to \cite[Lemma 9.44]{Besse-1987}, every horizontal geodesic $(\mathcal{M}_{<}(\hh),\mathcal{G}_{f})$ gives rise to a geodesic on $(\stsp_{md}(\hh),\GG_{f}^{\hh})$ when projected through $\pi$.
%%%%
A geodesic of $(\mathcal{M}_{<}(\hh),\mathcal{G}_{f})$ is horizontal  if and only if its tangent vector at each point lies in the orthogonal complement of the kernel of $T\pi$ at that point.
%%%%
In particular, this means that a geodesic is horizontal if and only if its tangent vector at $(\mathbf{U},\vec{p})$ does not have components along $\mathcal{X}_{lm}^{3}(\mathbf{U},\vec{p})$, where $\mathcal{X}_{lm}^{3}$ is the dual vector field to $\theta^{lm}_{3}$.
%%%%%

The fact that $\mathcal{G}_{f}$ splits into the sum of the Fisher-Rao metric tensor, plus a metric tensor on the special unitary group (which is a sort of weighted version of the Cartan-Killing metric tensor, with $\vec{p}$-dependent weights),  implies that the $\vec{p}$-component of the geodesic equation is completely independent from $\mathbf{U}$ (this may easily be seen by computing the Euler-Lagrange equation of the metric Lagrangian associated with $\mathcal{G}_{f}$ following, for instance, \cite{M-F-LV-M-R-1990}).
%%%%
Consequently, every geodesic $\gamma_{(\mathbf{U},\vec{p})}$ of $(\mathcal{M}_{<}(\hh),\mathcal{G}_{f})$ whose initial tangent vector  is  $(\mathbf{0},\vec{a})$ is of the type 
\be
\gamma_{(\mathbf{U},\vec{p})}(t)= (\mathbf{U},\vec{p}_{\vec{a}}(t)) 
\ee
where $\vec{p}_{\vec{a}}(t)$ is the geodesic of the Fisher-Rao metric tensor starting at $\vec{p}(0)=\vec{p}$ with initial tangent vector $\vec{a}$.
%%%%
Recalling that the Fisher-Rao metric tensor is the pull-back, through the square-root map, of the round metric tensor on an open portion of the n-sphere  to the n-simplex \cite{G-I-2001}, every geodesic $\vec{p}_{\vec{a}}(t)$ is given by
\be\label{eqn: universal geodesics}
p^{j}_{\vec{a}}(t)\,=\,\cos^{2}\left(\frac{t}{2}\,||\vec{a}||_{FR}\right)\,p^{j} + \frac{\sin^{2}\left(\frac{t}{2}\,||\vec{a}||_{FR}\right)}{||\vec{a}||^{2}_{FR}}\,\frac{a^{j}\,a^{j}}{p^{j}} + \frac{\sin\left(t\,||\vec{a}||_{FR}\right)}{||\vec{a}||_{FR}}\,a^{j}\,.
\ee
%%%%
Clearly, every  geodesic $\gamma_{(\mathbf{U},\vec{p})}$ is horizontal, and thus 
\be\label{eqn: universal geodesics 2}
\pi\circ\gamma_{(\mathbf{U},\vec{p})}(t)=\mathbf{U}\left(\sum_{j=1}^{n}p^{j}_{\vec{a}}(t)\mathbf{e}_{jj}\right)\mathbf{U}^{\dagger}
\ee
is a geodesic of $(\stsp_{md}(\hh),\GG_{f}^{\hh})$ for every operator monotone function $f$ characterizing the monotone Riemannian metric tensor $\GG_{f}^{\hh}$, i.e., it is a \grit{universal geodesic} for the family of monotone Riemannian metric tensors.
%%%%%

\section{Conclusions}\label{sec: conclusions}

In this contribution we analysed the Riemannian aspects of quantum information geometry from the point of view of the unfolding of quantum states. The latter were described  in terms of probability vectors and unitary operators built out of the spectral theorem, as described in section \ref{sec: unfolding of quantum states}.
%%%%%%

We argued that this point of view helps in making the comparison between classical and quantum information geometry more transparent. Indeed, we thoroughly discussed how every monotone Riemannian metric tensor $\GG_{f}^{\hh}$ falling in Petz's classification, when considered from the unfolded perspective on the unfolding space $\mathcal{M}(\hh):=\SUh\times\Delta_{n}$ introduced in equation \eqref{eqn: unfolding space 2}, splits into the sum of two contributions. The first one   is essentially the purely classical Fisher-Rao metric tensor, whereas the   second tensor is a kind of weighted sum of the Cartan-Killing form on $\SUh$,   whose weights depend on the operator monotone function characterising $\GG_{f}^{\hh}$ and on  points on the classical simplex $\Delta_{n}$ (see equation \eqref{eqn: unfolded Petz metric 6}).
%%%%
This leads us to conclude that all the possible monotone Riemannian metric tensors on the manifold $\stsp(\hh)$ of invertible quantum states share the same classical part.
%%%%%%
This instance, though conceptually clear in the standard formulation of quantum information geometry, becomes also “pictorially”  clear in the unfolded perspective thanks to equation  \eqref{eqn: unfolded Petz metric 6}.
%%%%%

Equation  \eqref{eqn: unfolded Petz metric 6} may be the starting point for studying other metric and curvature  properties of the monotone Riemannian metric tensors, exploiting the theory of Riemannian submersions.
%%%%
Indeed, in subsection \ref{subsec: universal geodesics} we already made a step in this direction by showing the existence of a family of geodesics which are common to all the monotone Riemannian metric tensors. They   are essentially the  geodesics of the classical Fisher-Rao metric, which are made quantum  by conjugation  with a fixed unitary operator (see equation \eqref{eqn: universal geodesics} and \eqref{eqn: universal geodesics 2}).
%%%%%

Of course, much more  could  be done by exploiting the unfolded perspective discussed here.
%%%%
For instance, it would be interesting to exploit the differential geometry of  $\mathcal{M}(\hh):=\SUh\times\Delta_{n}$, in order to compute the Riemann tensor and  the other curvature-related tensors of the unfolded metrics, so that, again exploiting the theory of Riemannian submersions, we can get insights on the corresponding objects on $\stsp(\hh)$.
%%%%

Given the splitting into classical and quantum part of the unfolded metrics, this perspective  will give a clear description of the classical contributions described in terms of the geometry of the Fisher-Rao metric tensor.
%%%%
Moreover, being  $\SUh$  a well-studied Lie group, this  may be of practical help in performing actual computations.
%%%%%

We also discussed how the unfolded perspective applies to the framework of relative g-entropies, and their associated Riemannian metric tensors.
%%%%
In particular, by exploiting the coordinate-free approach recalled in section \ref{sec: from 2-functions to covariant tensors}, we gave an alternative proof of the fact that relative g-entropies lead to monotone Riemannian metric tensors.
%%%%%
In this context, the obvious next step would be that of computing the family of dually-related connections determined by the relative g-entropies in the unfolded perspective.
%%%%
Our conjecture is that also the skewness tensor describing this dual structure will split into a purely classical part related to the geometry of the  Fisher-Rao metric tensor, and a quantum part along the special unitary group.
%%%%%

Another intriguing future direction of investigation would be that of testing the unfolded perspective in more applied contexts like, for instance, quantum estimation theory, quantum metrology, and quantum tomography where, hopefully,  working with probability vectors and unitary operators instead of positive operators may lead to practical advantages.
%%%%%%

In conclusion, we hope this work succeeds   to convey the message that the unfolded perspective here discussed actually presents some technical and theoretical features that deserve further study.
%%%%%

\section*{Acknoledgements}

\addcontentsline{toc}{section}{Acknowledgements}

The authors would like to thank Prof. G. Marmo for innumerable discussions and suggestions on the subjects of this work.
%%%%%
F. M. Ciaglia acknowledges that this work has been supported by the Madrid Government (Comunidad de Madrid-Spain) under the Multiannual Agreement with UC3M in the line of “Research Funds for Beatriz Galindo Fellowships” (C\&QIG-BG-CM-UC3M), and in the context of the V PRICIT (Regional Programme of Research and Technological Innovation).
%%%%
F. Di Cosmo thanks the UC3M, the European Commission through the Marie Sklodowska-Curie COFUND Action (H2020-MSCA-COFUND-2017-GA 801538) and Banco Santander for their financial support through the CONEX-Plus Program.
%%%%
P. Vitale acknowledges partial support by HPC – Centro Nazionale di Ricerca in High Performance Computing, Big Data and Quantum Computing, funded by European Union – NextGenerationEU. Her research is also supported by  Programme STAR Plus, financially supported by UniNA and Compagnia di San Paolo.

\addcontentsline{toc}{section}{References}

%\footnotesize{
%\bibliographystyle{plain}
%\bibliography{scientific_bibliography}
%}

\end{document}